\documentclass[pra,showpacs,graphics,twocolumn,floatfix,mathbbm,a4paper,nofootinbib]{revtex4-1}
\usepackage{amsthm}
\usepackage{amsmath}
\usepackage{latexsym}
\usepackage{amsfonts}
\usepackage{amssymb}
\usepackage{color}
\usepackage{bbm,dsfont}
\usepackage{graphicx}
\usepackage{subfigure}
\usepackage{mathbbol}
\usepackage{hyperref}
\usepackage{enumerate}
\usepackage{MnSymbol}

\newtheorem{proposition}{Proposition}

\newtheorem{corollary}{Corollary}

\theoremstyle{definition}

\newtheorem{example}{Example}
\newtheorem{definition}{Definition}

\newcommand{\real}{\mathbb R} 
\newcommand{\integer}{\mathbb Z} 
\newcommand{\half}{\tfrac{1}{2}} 

\newcommand{\hi}{\mathcal{H}} 
\newcommand{\no}[1]{\left\|#1\right\|} 
\newcommand{\tr}[1]{\mathrm{tr}\left[#1\right]} 
\newcommand{\id}{\mathbbm{1}} 

\newcommand{\ve}{\vec{e}} 
\newcommand{\vsigma}{\vec{\sigma}} 

\newcommand{\A}{\mathsf{A}}
\newcommand{\B}{\mathsf{B}}
\newcommand{\C}{\mathsf{C}}
\newcommand{\D}{\mathsf{D}}
\newcommand{\E}{\mathsf{E}}
\newcommand{\F}{\mathsf{F}}
\newcommand{\G}{\mathsf{G}}
\newcommand{\X}{\mathsf{X}}
\newcommand{\Y}{\mathsf{Y}}
\newcommand{\Z}{\mathsf{Z}}

\newcommand{\T}{\mathsf{T}}
\newcommand{\N}{\mathsf{N}}

\newcommand{\quant}{\mathcal{Q}} 

\newcommand{\state}{\mathcal{S}} 
\newcommand{\effect}{\mathcal{E}} 
\newcommand{\obs}{\mathcal{O}} 
\newcommand{\eff}{\mathcal{O}^{\textrm{eff}}} 
\newcommand{\noise}{\mathcal{N}} 
\newcommand{\trivial}{\mathcal{T}} 
\newcommand{\simu}[1]{\mathfrak{sim}(#1)} 
\newcommand{\smin}[1]{\mathfrak{s}_{\min}(#1)} 

\begin{document}

\title[]{Simulability of observables in general probabilistic theories}

\author{Sergey N. Filippov}
\email{sergey.filippov@phystech.edu}
\address{Institute of Physics and Technology, Russian Academy of Sciences, Moscow 117218, Russia}
\address{Moscow Institute of Physics and Technology, Dolgoprudny, Moscow Region 141700, Russia}

\author{Teiko Heinosaari}
\email{teiko.heinosaari@utu.fi}
\address{QTF Centre of Excellence, Department of Physics and Astronomy, University of Turku, Turku 20014, Finland}

\author{Leevi Lepp\"{a}j\"{a}rvi}
\email{leille@utu.fi}
\address{QTF Centre of Excellence, Department of Physics and Astronomy, University of Turku, Turku 20014, Finland}



\begin{abstract}
The existence of incompatibility is one of the most fundamental
features of quantum theory, and can be found at the core of many
of the theory's distinguishing features, such as Bell inequality
violations and the no-broadcasting theorem. A scheme for
obtaining new observables from existing ones via classical
operations, the so-called simulation of observables, has led to an
extension of the notion of compatibility for measurements. We
consider the simulation of observables within the operational
framework of general probabilistic theories and introduce the
concept of simulation irreducibility. While a simulation
irreducible observable can only be simulated by itself, we show
that any observable can be simulated by simulation irreducible
observables, which in the quantum case correspond to extreme rank-1 positive-operator-valued measures. We also consider cases where the set of
simulators is restricted in one of two ways: in terms of either
the number of simulating observables or their number of outcomes.
The former is seen to be closely connected to compatibility and
$k$--compatibility, whereas the latter leads to a partial
characterization for dichotomic observables. In addition to the
quantum case, we further demonstrate these concepts in state
spaces described by regular polygons.
\end{abstract}

\maketitle

\section{Introduction}

Recently, the concept of measurement simulability of quantum observables (modeled as positive-operator-valued measures) has been introduced and studied \cite{GuBaCuAc17,OsGuWiAc17}. 
It can be seen as a natural generalization of the concept of compatibility, and it allows one to study how one can implement a set of target observables from some chosen set of observables.
This kind of concept naturally arises in the studies of local hidden variable models \cite{HiQuVeNaBr17} as well as proposals to test fundamentally binary or $n$-ary theories \cite{KlCa16,KlVeCa17}.

The framework of general probabilistic theories (GPTs) is natural platform to investigate foundational aspects of quantum theory. Features of quantum theory, such as incompatibility and nonlocality, can be explored in a wider class of theories, allowing one to compare theories to one another and quantify how restricted these features are in different theories.
GPTs are based on operational notions of states and measurements so that, for example, an observable is any affine function that maps states into probability distributions. This is the exact analog of positive-operator-valued measures (POVMs) in the case of quantum theory.
The incompatibility of observables in GPTs has been recently studied in several works \cite{BuHeScSt13,StBu14,Banik15,Plavala16,FiHeLe17,JePl17}.
The purpose of the present paper is to formulate measurement simulability in the framework of GPTs and to further investigate the properties of this concept.

The difficulty or complexity of simulating a given collection of observables can be quantified by studying two types of limitations on the set of simulator observables.
First, we can look for the minimal set of simulator observables that can produce the target observables.
From this point of view, a target set is compatible if and only if it can be obtained with a single simulator observable.
Another quantification is obtained by allowing an arbitrary number of simulator observables but restricting them to have fewer outcomes than some threshold value.

We will demonstrate these two quantifications of simulability by comparing quantum theory to polygon theories \cite{JaGoBaBr11}.
It is interesting to recall that the so-called box world (i.e., square bit state space) \cite{PoRo94,GrMuCoDa10} possesses more incompatibility than any finite-dimensional quantum state space \cite{BuHeScSt13,HeScToZi14} if incompatibility is quantified as the global robustness under noise.
However, in both quantifications of simulability, the box world is closest to classical theory among all nonclassical theories.

The key concept in our investigation is \emph{simulation irreducibility}.
An observable has this property if it cannot be obtained from some essentially different simulator observables.
We present a general characterization of simulation irreducible observables and explicitly give them in several theories. In particular, we show that the set of all observables on state spaces described by regular polygons can be simulated by a finite number of trichotomic simulation irreducible observables with the only exception being the square bit state space, where simulation irreducible observables are dichotomic.

\section{Observables, postprocessing, and mixing} \label{sec:start}

\subsection{States, effects and observables} \label{sec:preliminaries}

We first recall some of the basic concepts of general probabilistic theories. The state space $\state$ is a compact convex subset of a finite-dimensional real vector space $\mathcal{V}$. The convexity arises from the probabilistic mixing of states so that for $p \in [0,1]$ and states $s_1,s_2 \in \state$ the convex sum $p s_1 + (1-p) s_2$ represents a state where we prepare the state $s_1 $ with probability $p $ and state $s_2$ with probability $1-p$.

An effect $e$ is given as a function $e: \state \to [0,1]$ on states such that
\begin{equation}\label{eq:affine}
e(ps_1 +(1-p)s_2) = p e(s_1) + (1-p) e(s_2).
\end{equation}
Then $e(s) \in [0,1]$ is interpreted as the probability that the measurement event that the effect $e$ represents happens when the system is in the state $s\in \state$. A functional $f: \state \to \real$ with property \eqref{eq:affine} is called affine on $\state$ and we denote by $F(\state)$ the set of affine functionals on $\state$. We can define a partial order in $F(\state)$ by denoting $e \leq f$ for $e,f \in F(\state)$ if $e(s) \leq f(s)$ for all $s \in \state$. The effect space can then be expressed as
\begin{equation}
\effect(\state) = \{ e \in F(\state) \, | \, o \leq e \leq u \},
\end{equation}
where $o$ and $u$ are the zero and unit effects respectively, i.e., $o(s) = 0$ and $u(s) =1$ for all $s \in \state$.

Sometimes it is useful consider the state space $\state$ as being embedded in an ordered vector space $\mathcal{A}$ such that $\state$ is a compact base for a generating positive cone $\mathcal{A}_+ = \{x \in \mathcal{A} \, | \, x \geq 0 \}$ \cite{CA70}. 
Hence, the state space can be expressed as
\begin{equation}
\state = \{ x \in \mathcal{A} \, | \, x \geq0, \ u(x) = 1\},
\end{equation}
i.e., as an intersection of the positive cone $\mathcal{A}_+$ and an affine hyperplane determined by (the extension of) the unit effect $u$ on $\mathcal{A}$. Furthermore, if $\dim({\rm aff}(\state))=d$, where ${\rm aff}(\state)$ denotes the affine span of $\state$, then we can take $\dim(\mathcal{A})=d+1$. It follows that, by adopting this approach, the effects can be expressed as linear functionals on $\mathcal{A}$ so that
\begin{equation}
\effect(\state) = \{ e \in \mathcal{A}^* \, | \, o \leq e \leq u \},
\end{equation}
where the partial order in the dual space $\mathcal{A}^*$ is the dual order
defined by the positive dual cone $\mathcal{A}^*_+ = \{ f \in
\mathcal{A}^* \, | \, f(x) \geq 0 \mathrm{\ for \ all\ } x \in
\mathcal{A}_+ \}$ of $\mathcal{A}_+$, and also $\dim(\mathcal{A}^*)=d+1$. In fact, $\effect(\state)$ is
then just the intersection of the positive dual cone
$\mathcal{A}^*_+$ and the set $u- \mathcal{A}^*_+$.

A nonzero effect $e$ is \emph{indecomposable} if a decomposition
$e=e_1+e_2$ is possible only when $e_1$ and $e_2$ are scalar
multiples of $e$; otherwise they are decomposable. It has been shown in
Ref. \cite{KiNuIm10} that in any GPT there exist indecomposable effects
and, further, any effect can be written as a finite sum of
indecomposable effects.  It is easy to see that the indecomposable
effects are exactly the ones laying on the extreme rays of the
cone $\mathcal{A}^*_+$.

Let $\state$ be a state space. An observable $\A$ with a finite number of outcomes is a map $\A: x \mapsto \A_x$ from a finite (outcome) set $X$ to
$\effect(\state)$ with the normalization $\sum_{x \in X}
\A_x(s) = 1$ for all $s \in \state$. The normalization condition, which is equivalent to the requirement that $\sum_{x \in X} \A_x = u$, guarantees that we detect with certainty one of the events corresponding to one of the effects $\A_x$ of the observable. We denote the set of all observables with outcome set $X$ by $\obs_X$ and the set of all observables with a finite number of outcomes on $\state$ by $\obs$.

An observable is called \emph{indecomposable} if all of its nonzero effects are indecomposable; otherwise it is \emph{decomposable}. From the decomposition of the unit effect into indecomposable effects, it follows that indecomposable observables do exist \cite{KiNuIm10}.

\begin{example}[\emph{Quantum theory}]
In finite-dimensional quantum theory the state space $\state_q$ is given by the set of positive trace-1 self-adjoint operators on a finite-dimensional Hilbert space $\hi$:
\begin{equation}
\state_q = \state(\hi) = \{ \varrho \in \mathcal{L}_s(\hi) \, | \, \varrho \geq O, \ \tr{\varrho}=1 \},
\end{equation}
where $\mathcal{L}_s(\hi)$ is the set of self-adjoint operators on $\hi$ and $O$ is the zero operator. The set of positive operators forms a generating positive cone in the vector space of self-adjoint operators $\mathcal{L}_s(\hi)$ with $\state(\hi)$ as its compact base. The effect space is given by the set of operators
\begin{equation}
 \effect(\hi) = \{ E \in \mathcal{L}_s(\hi) \, | \, O \leq E \leq \id \},
\end{equation}
where $\id$ is the identity operator, so that the one-to-one correspondence with the effect functionals in $\effect(\state_q)$ can be given by the equation $e(\varrho) = \tr{\varrho E}$. An observable $\A$ with a finite outcome set $X$ then corresponds to a POVM $A: x \mapsto A(x)$ such that $\sum_{x \in X} A(x) = \id$. An effect $E$ is indecomposable if and only if $E$ has rank equal to 1, or equivalently, $E$ is a scalar multiple of a one-dimensional projection \cite{KiNuIm10}.
\end{example}

\subsection{Postprocessing of observables}\label{sec:p-p}

A classical channel between outcome spaces $X$ and $Y$ is given by a (right) stochastic linear map $\nu: X \to Y$, i.e., map with matrix elements $\nu_{xy}$, $x \in X$, $y \in Y$ with $0 \leq \nu_{xy} \leq 1$ and $\sum_{x\in X} \nu_{xy} =1$. The matrix element $\nu_{xy}$ gives the transition probability that outcome $x$ is mapped into outcome $y$. In addition to being used as a transformation between outcome spaces, classical channels are most commonly used to describe noise.

For an observable $\A$ with an outcome set $X$ and a classical channel $\nu: X \to Y$ between $X$ and some other outcome space $Y$ we denote by
$\nu\circ\A$ a new observable defined as
\begin{align}
(\nu \circ \A)_y = \sum_{x \in X} \nu_{xy} \A_x
\end{align}
for all outcomes $y \in Y$.
Physically, the observable $\nu\circ\A$ can be implemented by first measuring $\A$ and then using the classical channel $\nu$ on each measurement outcome.

For two observables $\A$ and $\B$, we say that \emph{$\B$ is
a postprocessing of $\A$}, denoted by $\A\to\B$, if there exists a classical channel
$\nu$ such that $\B = \nu \circ \A$.
In the context of quantum observables, this relation was introduced in Ref. \cite{MaMu90a}.
We follow the terminology of Ref. \cite{BuDaKePeWe05} and say that an observable $\A$ is \emph{postprocessing clean} if, for any observable $\B$, the relation $\B\to\A$ implies that $\A \to \B$. We have the following characterization:

\begin{proposition}\label{prop:pp-clean}
An observable is postprocessing clean if and only if it is indecomposable.
\end{proposition}

\begin{proof}
Let $\A$ be a postprocessing clean observable with an outcome set $\Omega$.
In Ref. \cite{KiNuIm10} it was shown that each nonzero effect $\A_x$ has a decomposition $\A_x= \sum_{i=1}^{r_x} a^{(x)}_i$ into $r_x<\infty$ indecomposable effects $a^{(x)}_i$.
We denote $r = \max_{x \in \Omega} r_x$ and define an observable $\B$ with an outcome set $\{1,\ldots,r\} \times \Omega$ by $\B_{(i,x)} = a^{(x)}_i$ if $i \leq r_x$ and $\B_{(i,x)}=o$ otherwise. Now we see that
\begin{align*}
\A_x &= \sum_{i=1}^{r_x} a^{(x)}_i = \sum_{i=1}^r \B_{(i,x)} = \sum_{i,x'} \nu_{(i,x')x} \B_{(i,x')},
\end{align*}
where we have defined the postprocessing $\nu: \{1,\ldots,r\}
\times \Omega \to \Omega$ by $\nu_{(i,x')x}= \delta_{x',x}$ for
all $i= 1, \ldots,r$ and $x,x' \in \Omega$. Thus, $\B \to \A$.
Since $\A$ is postprocessing clean, it follows that also $\A \to
\B$, hence there exists a postprocessing $\mu: \Omega \to
\{1,\ldots,r\} \times \Omega$ such that
\begin{equation*}
\B_{(i,x)} = \sum_{y \in \Omega} \mu_{y(i,x)} \A_{y}
\end{equation*}
for all $i=1,\ldots,r$ and $x \in \Omega$.
Each nonzero effect $\B_{(i,x)} =a^{(x)}_i$ is indecomposable, and so for all $x,y \in \Omega$ and $i=1, \ldots,r_x$ there exists a real number $p^i_{xy} >0$ such that $\mu_{y(i,x)} \A_{y}= p^i_{xy} a^{(x)}_i$.
From the normalization $\sum_{i,x} \mu_{y(i,x)}=1$ for all $y \in \Omega$ it follows that for each $y \in \Omega$ (such that $\A_y \neq o$) there exists an element $\mu_{y(i_y,x_y)} \neq 0$ for some $i_y=1, \ldots, r_{x_y}$ and $x_y \in \Omega$.
Hence, for each $y \in \Omega$ (such that $\A_y \neq o$), there exists an outcome $(i_y,x_y)$ for the observable $\B$ with an indecomposable effect $\B_{(i_y,x_y)}=a^{(x_y)}_{i_y} $ such that
\begin{equation*}
\A_y = \frac{p^{i_y}_{x_yy}}{\mu_{y(i_y,x_y)}} a^{(x_y)}_{i_y}.
\end{equation*}
Thus, each nonzero effect of $\A$ is indecomposable.

Let then $\A$ be an indecomposable observable with an outcome set $\Omega$.
We consider an observable $\C$ with an outcome set $\Omega'$ such that $\C \to \A$, i.e., there exists a postprocessing $\eta: \Omega' \to \Omega$ such that $\A_x = \sum_{z \in \Omega'} \eta_{zx} \C_z$ for all $x \in \Omega$. Without loss of generality observable $\C$ has only nonzero outcomes. Thus, each effect $\C_z$ has a decomposition $\C_z = \sum_{j=1}^{r_z} c^{(z)}_j$ into $r_z$ indecomposable effects $c^{(z)}_j$, and
\begin{equation*}
\A_x = \sum_{z \in \Omega'} \sum_{j=1}^{r_z} \eta_{zx} c^{(z)}_j
\end{equation*}
for all $x \in \Omega$.
Hence, for each $z \in \Omega'$, $j\in \{1, \ldots,r_z\}$ and $x \in \Omega$ such that $\A_x \neq o$ there exists a real number $q^j_{xz} \geq 0$ such that $\eta_{zx} c^{(z)}_j = q^j_{xz} \A_x$.
By summing over all $z \in \Omega'$ and $j=1,\ldots,r_z$ we have that
\begin{equation*}
\left(\sum_{z \in \Omega'} \sum_{j=1}^{r_z} q^j_{xz} \right) \A_x = \sum_{z \in \Omega'} \sum_{j=1}^{r_z} \eta_{zx} c^{(z)}_j = \sum_{z\in \Omega'} \eta_{zx} \C_z = \A_x
\end{equation*}
for all $x \in \Omega$ such that $\A_x \neq o$.
Thus, for such $x \in \Omega$ we have that $\sum_{z \in \Omega'} \sum_{j=1}^{r_z} q^j_{xz} =1$. From this it also follows that $0\leq\sum_j q^j_{xz}\leq 1$ for all $z \in \Omega'$ and $x \in \Omega$ such that $\A_x \neq o$ .

Since $\eta_{zx} = 0$ for all $x \in \Omega$ such that $\A_x=o$, we have from the normalization of the postprocessing $\eta$ that
\begin{equation*}
c^{(z)}_j = \sum_{x \in \Omega} \eta_{zx} c^{(z)}_j =  \sum_{\substack{x \in \Omega \\ \A_x \neq o}} \eta_{zx} c^{(z)}_j+ \sum_{\substack{x \in \Omega \\ \A_x = o}} \eta_{zx} c^{(z)}_j = \sum_{\substack{x \in \Omega \\ \A_x \neq o}} q^j_{xz} \A_x.
\end{equation*}
Thus,
\begin{equation*}
\C_z= \sum_{j=1}^{r_z} c^{(z)}_j = \sum_j \sum_{\substack{x \in \Omega \\ \A_x \neq o}} q^j_{xz} \A_x = \sum_{x \in \Omega} \lambda_{xz} \A_x,
\end{equation*}
where we have defined $\lambda_{xz} = \sum_{j=1}^{r_z} q^j_{xz}$ when $\A_x \neq o$ and $\lambda_{xz} = 1/ \#\Omega' $ otherwise. From the observations made above we have that $0 \leq \lambda_{xz} \leq 1$ for all $x\in \Omega$ and $z \in \Omega'$, and furthermore $\sum_{z \in \Omega'} \lambda_{xz} =1$ for all $x \in \Omega$ so that the map $\lambda: \Omega \to \Omega'$ defined by matrix element $\lambda_{xz}$ is a postprocessing.
Hence, $\A \to \C$ for all observables $\C$ such that $\C \to \A$ and so $\A$ is postprocessing clean.
\end{proof}

The postprocessing relation is a preorder on $\obs$, i.e., a
transitive and symmetric relation. Two observables $\A$ and $\B$
are \emph{postprocessing equivalent} if both $\A\to\B$ and
$\B\to\A$, and in this case we denote $\A \leftrightarrow \B$.
This is an
equivalence relation, and the set $\obs$ therefore splits into
equivalence classes.
Two postprocessing equivalent observables do not differ in any physically relevant way.

\begin{example}[\emph{Minimally sufficient representative}]
In every equivalence class, one has an observable for which all effects are pairwise linearly independent. This was proven for quantum observables with a finite number of outcomes in Ref. \cite{MaMu90a}. A generalization of this property was introduced and studied in Ref. \cite{Kuramochi15b}, where such observables were called \emph{minimally sufficient}.

To see that an observable with pairwise linearly independent
effects exists for each equivalence class in our setting, let us
consider an observable $\A: X \to \effect(\state)$. Suppose that
two effects $\A_{x'}$ and $\A_{x''}$ are linearly dependent
(proportional to each other). Consider the outcome set
$Y=X \setminus \{x''\}$ and a postprocessing $\nu: X \to Y$ such
that $\nu_{xy} = \delta_{xx'}+\delta_{xx''}$ if $y=x'$ and
$\nu_{xy} = \delta_{xy}$ otherwise. In the resulting observable
$\B= \nu \circ \A$, the effects $\A_{x'}$ and $\A_{x''}$ are
merged into $\B_{x'} = \A_{x'} + \A_{x''}$. Thus, by construction
$\A\to\B$.

Note that $\A_{x'}=p'\B_{x'}$, $\A_{x''}=p''\B_{x'}$, where
$p',p'' \geq 0$ and $p'+p''=1$. By defining the postprocessing
$\mu: Y \to X$ such that $\mu_{yx} =
p'\delta_{xx'}+p''\delta_{xx''}$ if $y=x'$ and $\mu_{yx} =
\delta_{xy}$ otherwise, we see that $\A =  \mu \circ \B$, and
hence $\B \to \A$. Therefore, $\A\leftrightarrow \B$. By
continuing this kind of merging of linearly dependent pairs of
effects, we will eventually obtain an observable $\hat{\A}$ with
pairwise linearly independent effects which is postprocessing
equivalent with $\A$.

Furthermore, it can be shown that the observable $\hat{\A}$ is
essentially unique: If $\widetilde{\A}$ is another pairwise
linearly independent observable in the equivalence class of $\A$,
then the postprocessing equivalence between $\hat{\A}$ and
$\widetilde{\A}$ is given by permutation matrices so that the
observables are only bijective relabellings of each other. In
\cite{MaMu90a} this was proved for quantum observables but since
the proof is analogous in the GPT framework it is omitted here.
\end{example}

\subsection{Mixing of observables}

A mixing of observables means a procedure where, in each
measurement round, we randomly pick an observable from a finite
collection and measure it. Thus, if we have $m$ observables
$\B^{(1)}, \ldots, \B^{(m)}$ with respective outcome sets
$X_1, \ldots,X_m$, then for any
probability distribution $p: i \mapsto p_i$ on $\{1, \ldots,m\}$
we can form an observable $\B$ with the outcome set $X \equiv
\cup_{i=1}^m X_i$ by
\begin{equation}
\B_x = \sum_{i=1}^m p_i \B^{(i)}_x \, ,
\end{equation}
where each observable $\B^{(i)}$ is extended onto $X$ by setting $\B^{(i)}_x=o$ if $x \notin X_i$. We can therefore assume that the outcome sets of the mixed observables are the same.

There is another way of forming mixtures, found by keeping
track of the measured observable in each round of the measurement.
Just as above, we take the outcome sets of the observables
$\B^{(1)}, \ldots, \B^{(m)}$ to be equal, say
$X$, but now the outcome in each measurement round is a pair
$(k,x)$, where $k \in \{1, \ldots,m\}$ labels the measured
observable $\B^{(k)}$ and $x \in X$ is the obtained outcome. To
formulate the mixing procedure mathematically, the mixture of
these observables is an observable $\tilde{\B}$ with the outcome
set $X_0=\{1, \ldots,m\} \times X$ defined as
\begin{equation}
\tilde{\B}_{(k,x)} = p_k \B^{(k)}_{x}
\end{equation}
for each $k \in \{1,\ldots,m\}$ and $x \in X$, where $p_k$ is the probability of measuring the observable $\B^{(k)}$.

It is clear that the latter way of mixing leads to a finer
observable than the first one; by postprocessing $\tilde{\B}$ we
can obtain the observable that corresponds to mixing without
keeping track of the measured observable. Namely, we define a
function $f:X_0 \to X$ by $f(k,x)=x$ and define the relabelling
$\nu^f: X_0 \to X$ by $\nu^f_{(k,x)y}= 1$ if $f(k,x)=x=y$ and
$\nu^f_{(k,x)y}=0$ if $f(k,x)=x \neq y$. Then
\begin{equation}
(\nu^f \circ \tilde{\B})_x = \sum_{i,x'} \nu^f_{(i,x')x} p_i
\B^{(i)}_{x'} = \sum_i p_i \B^{(i)}_{x} = \B_x \, .
\end{equation}

In the following definition we will understand mixing in the sense that the outcome sets are the same and we do not keep track of the mixed observables.

\begin{definition}
\label{def:extreme} An observable $\A\in\obs_X$ is called
\emph{extreme} if a convex sum decomposition $\A = \lambda \B +
(1-\lambda) \C$, $0 < \lambda < 1$ with $\B,\C \in \obs_X$ implies
$\B=\C=\A$.
\end{definition}

The following is a well-known fact for quantum observables; see, e.g., Ref.  \cite{HaHePe12}.
It is proven analogously in the GPT framework and a proof is given in the appendix.

\begin{proposition}\label{prop:nonzero-extreme}
Nonzero effects of an extreme
observable are linearly independent.
\end{proposition}

Also the following statement is well-known for quantum observables; see, e.g., Ref. \cite{Parthasarathy99}.
We find it useful to give a proof that is valid in the GPT framework.

\begin{proposition}\label{prop:p-p-clean-extreme}
A postprocessing clean observable
$\A$ is extreme if and only if its nonzero effects are linearly
independent.
\end{proposition}

\begin{proof}
The necessity of linear independence follows from
Proposition~\ref{prop:nonzero-extreme}. To prove sufficiency, let the
nonzero effects $\A_1,\ldots,\A_n$ of a
postprocessing clean observable $\A$ be linearly independent. Let
$\A = \sum_{k} p_k \B^{(k)}$ for some probability distribution
$\{p_k\}_k$ and some set of observables $\{\B^{(k)}\}_k$. We
define an observable $\widetilde{\B}$ as  $\widetilde{\B}_{(k,x)}
= p_k \B_x^{(k)}$. Then
\begin{equation}
\label{eq:B-to-A} \A_y = \sum_{k,z} \nu_{(k,z)y}
\widetilde{\B}_{(k,z)},
\end{equation}
where the postprocessing matrix $\nu$ has the
form
\begin{equation}
\label{eq:nu-form} \nu_{(k,z)y} = \left\{ \begin{array}{ll}
  1 & \text{if~} z=y, \\
  0 & \text{if~} z \neq y. \\
\end{array} \right.
\end{equation}
Thus, $\widetilde{\B} \to \A$. Since $\A$ is
postprocessing clean, the latter relation implies $\A \to
\widetilde{\B}$; i.e., there exists a postprocessing matrix
$\mu_{x(k,z)}$ such that
\begin{equation}
\label{eq:A-to-B} \widetilde{\B}_{(k,z)} = \sum_x \mu_{x(k,z)} \A_x.
\end{equation}
Combining \eqref{eq:B-to-A} with \eqref{eq:A-to-B}, we get
\begin{equation}
\A_y = \sum_{x} \left( \sum_{k,z} \mu_{x(k,z)} \nu_{(k,z)y}
\right) \A_x.
\end{equation}
Since the effects $\A_1,\ldots,\A_n$ are linearly
independent, the term in parentheses must be equal to
$\delta_{xy}$.
Further, since $\nu$ has the specific form of \eqref{eq:nu-form}, we have that $\sum_{k} \mu_{x(k,x)} = 1$ for all $x$. Since
$\sum_{k,z} \mu_{x(k,z)} = 1$ for all $x$ and $k$, it follows that all the elements
$\mu_{x(k,z)}$ with $z \neq x$ are zero. Thus, only one term in
the sum \eqref{eq:A-to-B} contributes and hence
$$
p_k \B_z^{(k)} =\widetilde{\B}_{(k,z)} = \mu_{z(k,z)} \A_z \, ,
$$
i.e., $\B_z^{(k)} =
\frac{\mu_{z(k,z)}}{p_k} \A_z$. Finally, as $\sum_z \B_z^{(k)} =
u$ and $\sum_z \A_z = u$, we get
\begin{equation}
\sum_z \left( \B_z^{(k)} - \A_z \right) = \sum_z \left(
\frac{\mu_{z(k,z)}}{p_k} - 1 \right) \A_z = 0.
\end{equation}
Since the effects $\A_1,\ldots,\A_n$ are linearly independent, the
latter equation implies $\mu_{z(k,z)} = p_k$ for all $z$, so
$\B_z^{(k)} = \A_z$ and $\B^{(k)} = \A$ for all $k$.
This means that $\A$ is extreme.
\end{proof}

\section{Simulation of observables} \label{sec:simulation}

\subsection{Simulation scheme}

\begin{figure}
\includegraphics[scale=0.45]{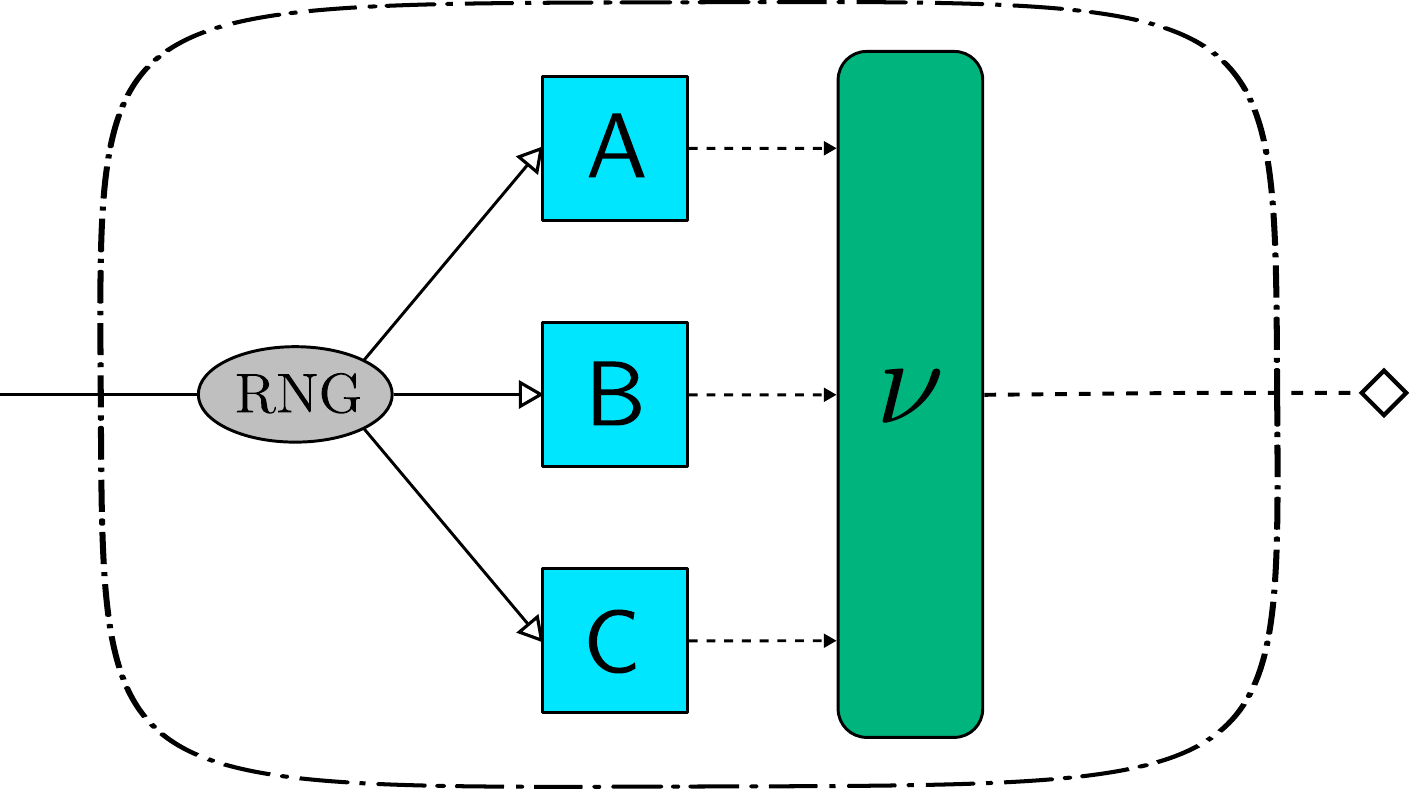}
\caption{\label{figure1} Simulation of a new observable by using
observables $\A$, $\B$, and $\C$. In every individual measurement,
the random number generator (RNG) chooses one of observables $\A$,
$\B$, and $\C$. Then, the classical outcome is affected by a
classical channel $\nu$.}
\end{figure}

Let us consider a subset $\mathcal{B} \subseteq \mathcal{O}$ of observables.
Following Ref. \cite{GuBaCuAc17}, we consider the set of observables that can be obtained from $\mathcal{B}$ by means of classical manipulations, namely by mixing and postprocessing. The simulation scheme consists of two steps: i) for any finite subset $\{\B^{(i)} \}_{i=1}^m \subseteq \mathcal{B}$ of observables with outcome set $X$ we choose an observable $\B^{(i)}$ with some probability $p_i$ and measure it, and ii) after obtaining an outcome $(i,x)$ by keeping track of the measured observable we perform some postprocessing $\nu: \cup_{k=1}^m \{k\} \times X \to Y$, outputting an outcome $y\in Y$ for some outcome space $Y$ with a probability $\nu_{(i,x)y}$.
Thus, the result is an observable $\A$ with an outcome set $Y$ such that
\begin{equation}\label{eq:sim_obs}
\A_y = (\nu \circ \tilde{\B})_y = \sum_{(i,x)} \nu_{(i,x)y} \tilde{\B}_{(i,x)}
\end{equation}
for all $y \in Y$, where $\tilde{\B}_{(i,x)} = p_i \B^{(i)}_x$ is
the observable used to define the mixture where we keep track of
the outcomes. The scheme is depicted in Fig.~\ref{figure1}.

We see that we can write $\A$ in two equivalent ways. By expanding \eqref{eq:sim_obs}, we see that
\begin{align}
\begin{split}
\A_y &= \sum_{(i,x)} \nu_{(i,x)y} \tilde{\B}_{(i,x)} =\sum_{(i,x)} \nu_{(i,x)y}  p_i \B^{(i)}_x \\
&= \sum_i p_i \left( \sum_x \nu_{(i,x)y} \B^{(i)}_x \right).
\end{split}
\end{align}
Now we may split the postprocessing $\nu$ into $m$ parts by defining postprocessings $\nu^{(i)}: X \to Y$ by
\begin{equation}\label{eq:pp-split}
\nu^{(i)}_{xy} = \nu_{(i,x)y}
\end{equation}
for all $x \in X$, $y \in Y$ and $i= 1, \ldots,m$. Hence, we can express $\A$ as
\begin{equation}\label{eq:pp-split}
\A_y = \sum_i p_i \left( \sum_x \nu^{(i)}_{xy} \B^{(i)}_x \right) = \sum_i p_i ( \nu^{(i)} \circ \B^{(i)})_y.
\end{equation}
Thus, we can think of either first mixing the observables and then
postprocessing the mixture, or first postprocessing the
observables individually and then mixing the post-processed
observables. The scheme in Fig.~\ref{figure1} is therefore equivalent
to the scheme in Fig.~\ref{figure2}.

\begin{figure}
\includegraphics[scale=0.45]{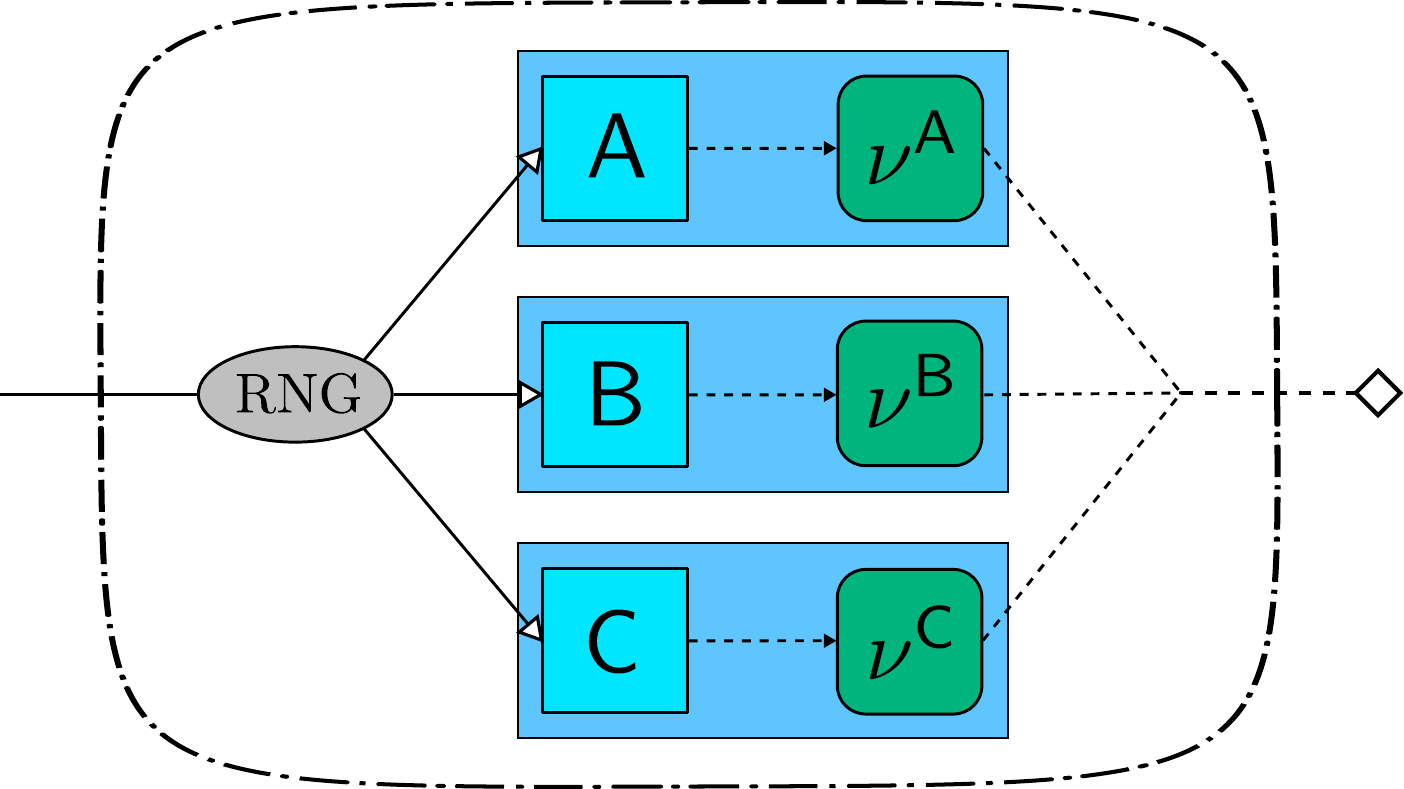}
\caption{\label{figure2} Simulation can be equivalently seen as a mixture of post-processed observables.}
\end{figure}

As an additional remark, let us consider the case where some of the observables used in the simulation are the same. Suppose we have an observable $\A$ with an outcome set $Y$ that can be simulated by observables $\B^{(1)}, \ldots, \B^{(m)}, \B^{(m+1)}, \ldots, \B^{(n)} $ with outcome sets $X$ such that $\B^{(m+1)}= \B^{(m+2)}=\cdots=\B^{(n)}$, i.e., we can express $\A$ as
\begin{align}
\begin{split}
\A_y &= \sum_{(i,x)} \nu_{(i,x)y}p_i \B^{(i)}_x \\
&= \sum_{i=1}^m \sum_x \nu_{(i,x)y}p_i \B^{(i)}_x + \sum_{i=m+1}^n \sum_x \nu_{(i,x)y}p_i \B^{(m+1)}_x
\end{split}
\end{align}
for all $y \in Y$ with some probability distribution $(p_i)_{i=1}^n$ and postprocessing $\nu: \cup_{k=1}^n \{k\} \times X \to Y$.
We can then form a new probability distribution $(p'_i)_{i=1}^{m+1}$ by   \begin{align}
\begin{split}
& p'_i = p_i \quad \forall i=1, \ldots, m, \\
& p'_{m+1} = \sum_{j=m+1}^n p_j.
\end{split}
\end{align}
and a new postprocessing $\nu': \cup_{k=1}^{m+1} \{k\} \times X \to Y$ by setting
\begin{align}
\begin{split}
& \nu'_{(i,x)y} = \nu_{(i,x)y} \quad \forall i=1, \ldots, m, \\
& \nu'_{(m+1,x)y} = \sum_{j=m+1}^n \dfrac{p_j}{p'_{m+1}} \nu_{(j,x)y}
\end{split}
\end{align}
for all $x \in X$ and $y \in Y$, so that
\begin{eqnarray*}
\A_y
&=& \sum_{i=1}^m \sum_x \nu_{(i,x)y}p_i \B^{(i)}_x + \sum_x p'_{m+1}\sum_{i=m+1}^n  \dfrac{p_i}{p'_{m+1}}\nu_{(i,x)y} \B^{(m+1)}_x \\
&=& \sum_{i=1}^m \sum_x \nu'_{(i,x)y}p'_i \B^{(i)}_x + \sum_x \nu'_{(m+1,x)y}p'_{m+1} \B^{(m+1)}_x
\\
&=&  \sum_{i=1}^{m+1} \sum_{x} \nu'_{(i,x)y}p'_i \B^{(i)}_x
\end{eqnarray*}
for all $y\in Y$. Hence, instead of using multiple instances of
the same observable in the simulating scheme, by modifying the
mixing and postprocessing we can reduce the multiplicity so that
only one instance of each different observable is used. The
intuitive reason for this is that when there are multiple
instances of the same observable in the simulator, a \emph{router}
can be used to direct the outcomes to the individual
postprocessings with some (weighted) probabilities resulting in a
reduction of multiplicity; see Fig.~\ref{figure3}. Looking from
the other way round, we can think of using the same simulator
observable several times, even if we would have only a single
device to hand.

\begin{figure}
\includegraphics[scale=0.296]{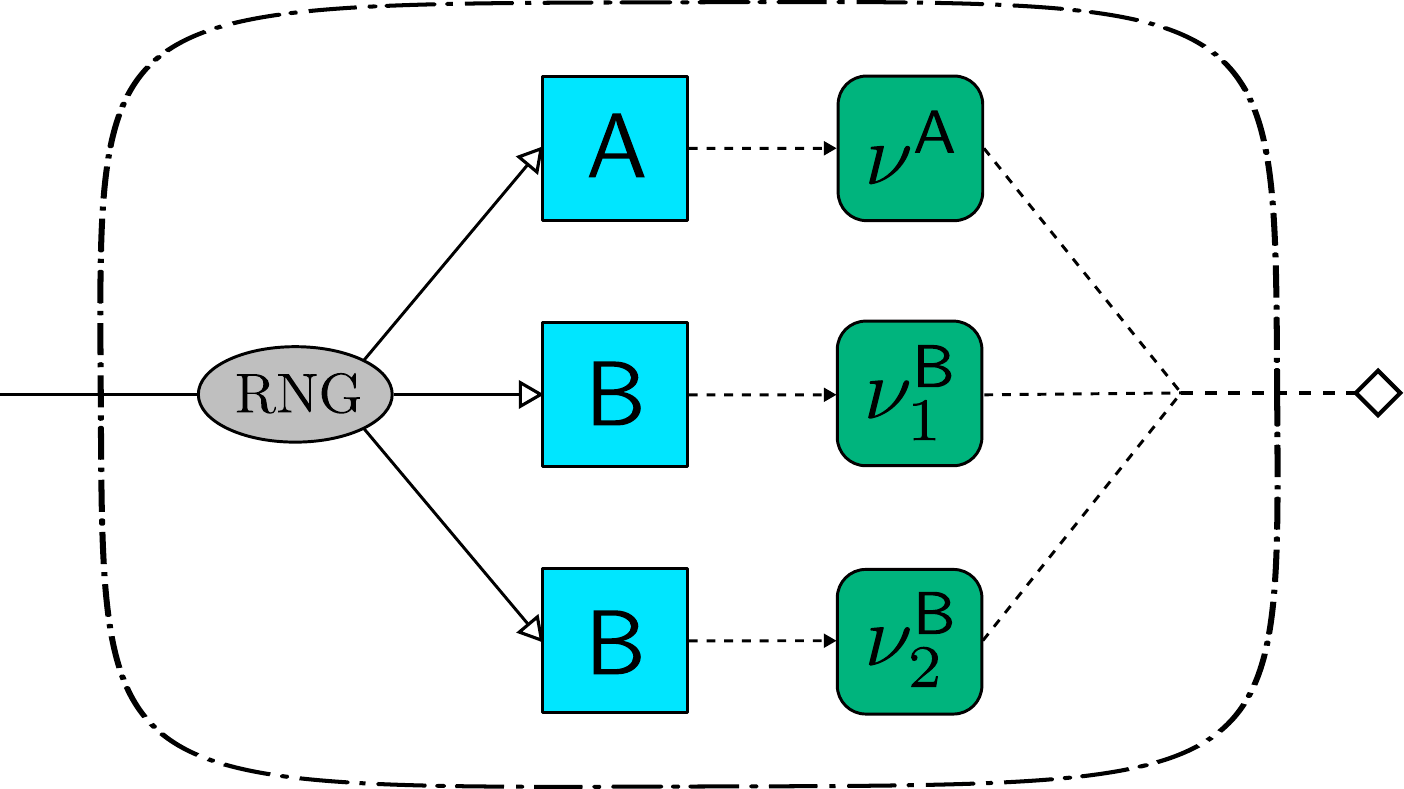} \
\includegraphics[scale=0.296]{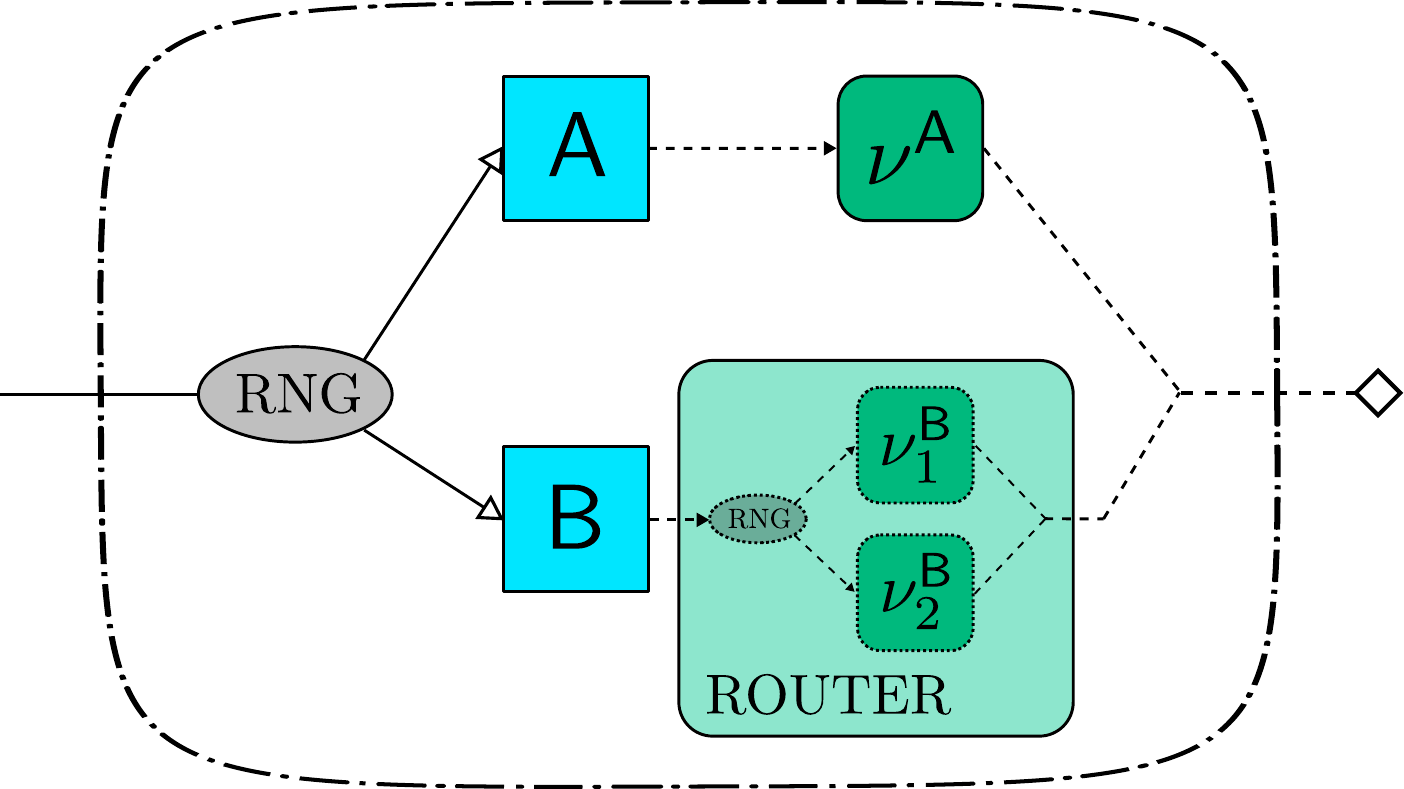}
\caption{\label{figure3} Multiple uses of a same observables are allowed but do not alter the generality of simulability. The scheme on the left hand side can be reduced to the scheme on the right hand side.}
\end{figure}

\subsection{The simulation map}

 Consider a subset of observables $\mathcal{B} \subseteq \obs$. Following the terminlogy from Ref. \cite{GuBaCuAc17}, we say that an observable $\A$ is \emph{$\mathcal{B}$-simulable} if it can be implemented with a simulation scheme by using some finite number of observables from $\mathcal{B}$.
Further, we denote by $\simu{\mathcal{B}}$ the set of all observables that are $\mathcal{B}$-simulable, and we treat $\simu{\cdot}$ as a map on the power set $2^\obs$.
In the case of a singleton set $\{ \B \}$, we simply denote $\simu{\B} \equiv \simu{ \{ \B \} }$.

For any subsets $\mathcal{B},\mathcal{C}\subseteq\obs$, the map $\simu{\cdot}$ satisfies the following basic properties:
\begin{itemize}
{\setlength\itemindent{25pt}
\item[(sim1)] $\mathcal{B} \subseteq \simu{\mathcal{B}}$,
\item[(sim2)] $\simu{\simu{\mathcal{B}}}=\simu{\mathcal{B}}$,
\item[(sim3)] $\mathcal{B} \subseteq \mathcal{C} \Rightarrow \simu{\mathcal{B}}\subseteq \simu{\mathcal{C}}$.
}
\end{itemize}
These properties are easy to verify and they mean that $\simu{\cdot}$ is a \emph{closure operator} on $\obs$.
It is commonly known that the closure operator properties (sim1)--(sim3) are equivalent to the single condition:
\begin{itemize}
{\setlength\itemindent{25pt}
\item[(sim4)] $\mathcal{B} \subseteq \simu{\mathcal{C}} \Leftrightarrow \simu{\mathcal{B}} \subseteq \simu{\mathcal{C}}$.
}
\end{itemize}

In the definition of simulability we are requiring that the
simulation scheme consists of a finite number of observables. It
thus follows that
\begin{itemize}
{\setlength\itemindent{25pt}
\item[(sim5)] $\simu{\mathcal{B}} = \bigcup \{ \simu{\mathcal{B}'}: \mathcal{B}'\subseteq\mathcal{B}\textrm{ and $\mathcal{B}'$ is finite}\}$.
}
\end{itemize}
This property means that $\simu{\cdot}$ is an \emph{algebraic} closure operator.

The map $\simu{\cdot}$ also has the following two properties:
\begin{itemize}
{\setlength\itemindent{25pt}
\item[(sim6)] $\simu{\mathcal{B}}$ is convex, i.e., closed under mixing,
\item[(sim7)] $\simu{\mathcal{B}}$ is closed under postprocessing.
}
\end{itemize}
The properties (sim6) and (sim7) are straightforward to verify by
noticing the equivalent ways to write mixtures and
postprocessings; see Figs. \ref{figure4} and \ref{figure5}.
Complete proofs are presented in the appendix.

\begin{figure}
\includegraphics[scale=0.16]{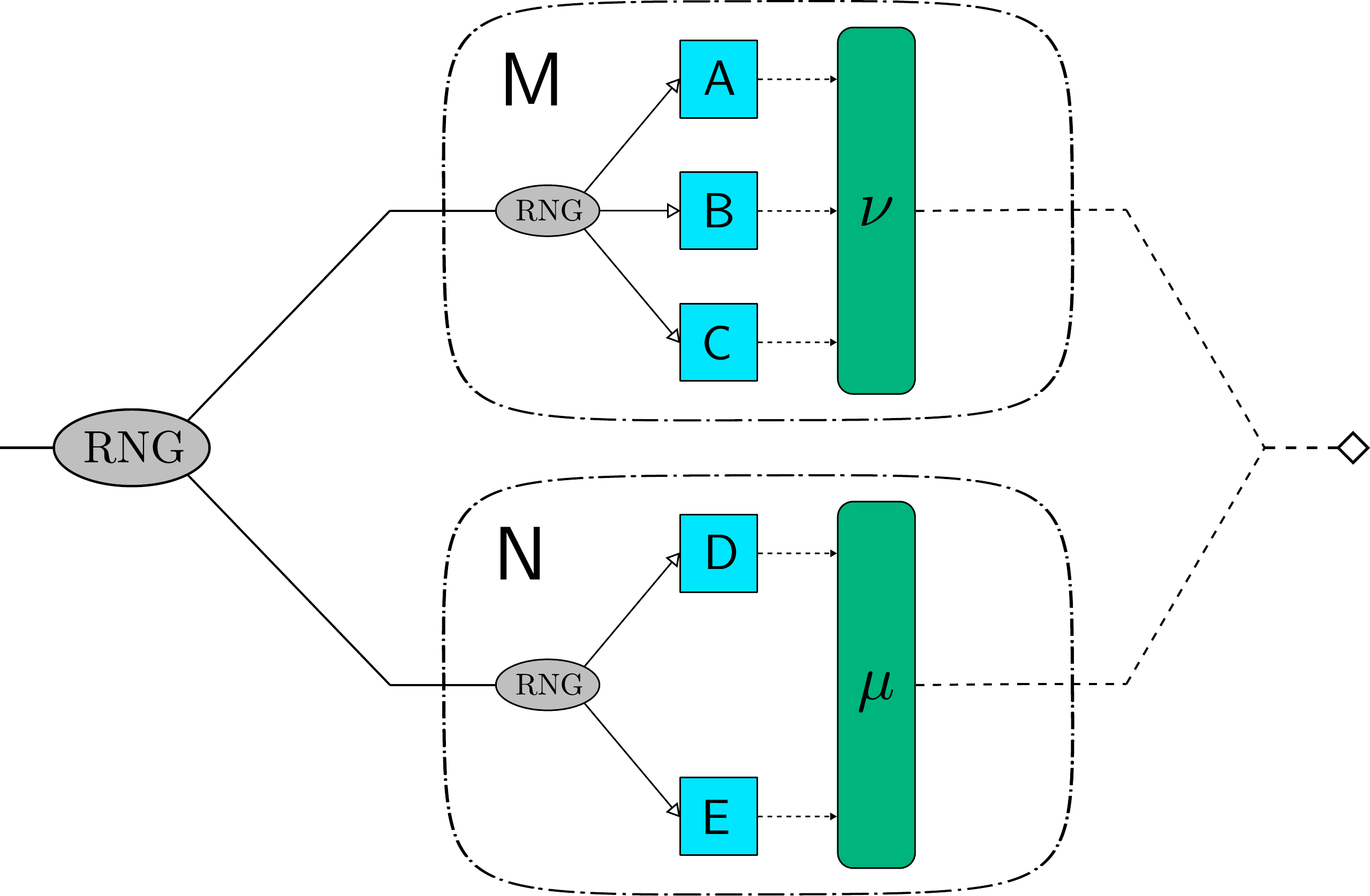} \
\includegraphics[scale=0.29]{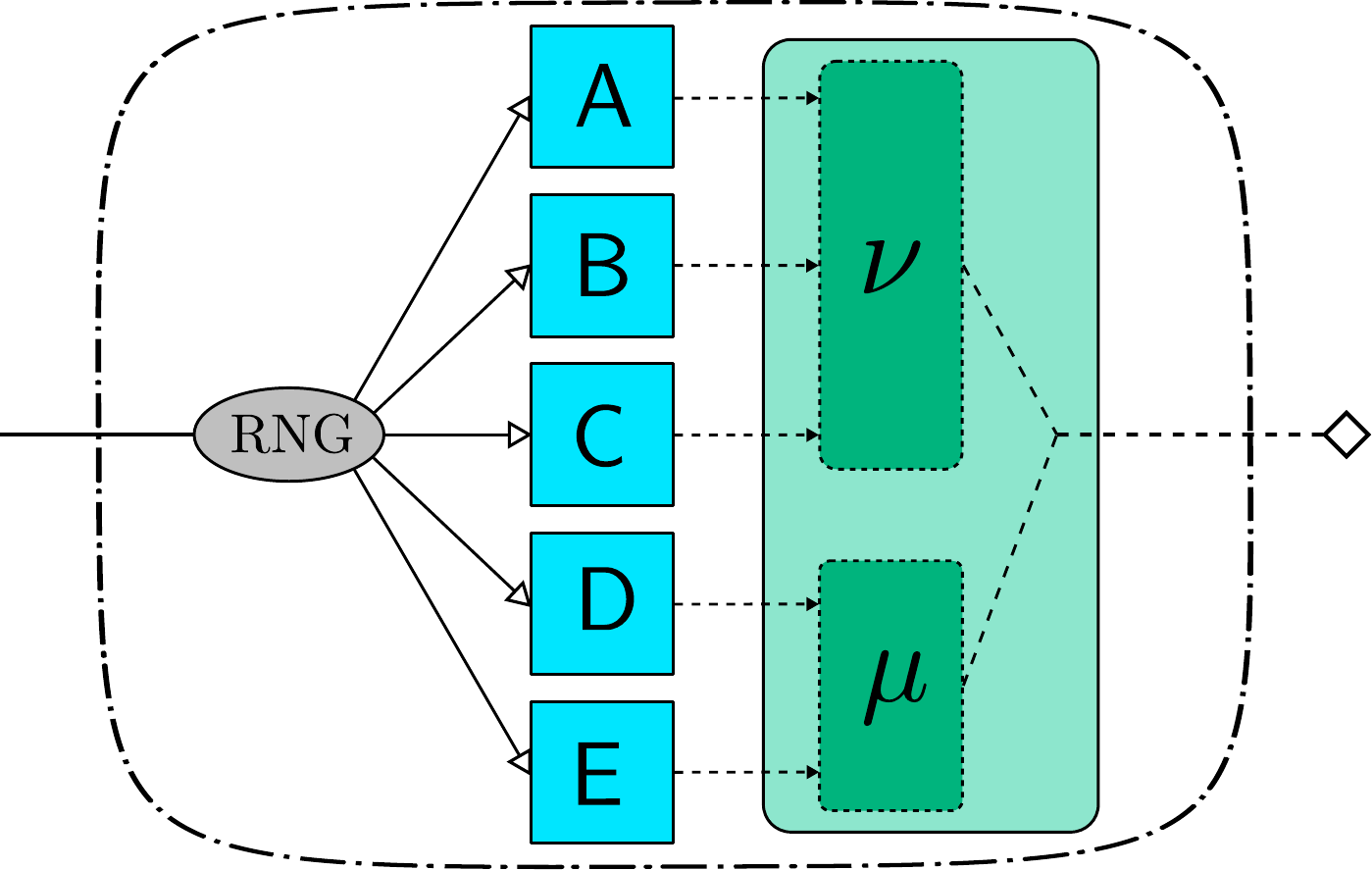}
\caption{\label{figure4} The set of simulable observables is
convex.}
\end{figure}

\begin{figure}
\includegraphics[scale=0.26]{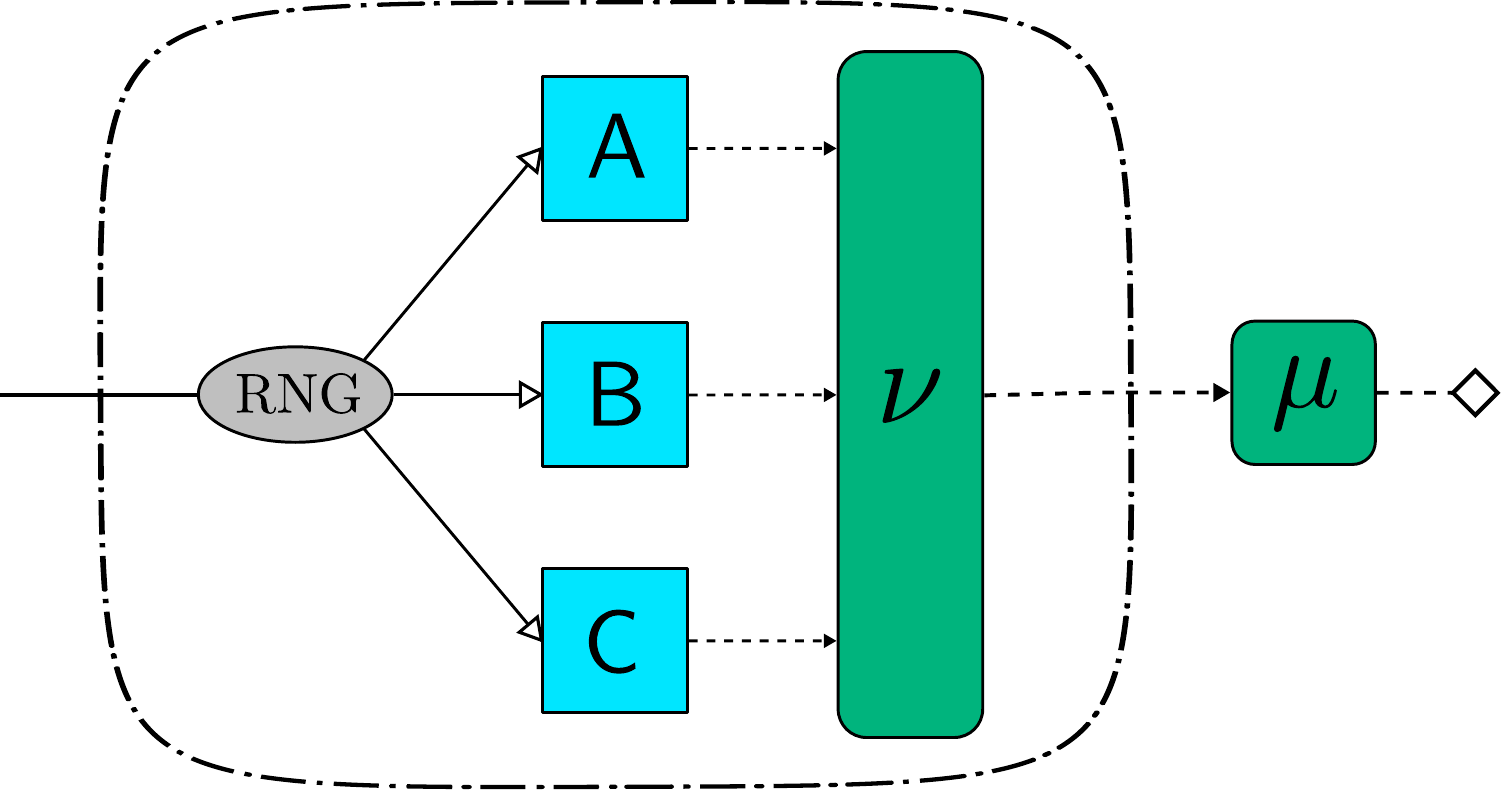} \
\includegraphics[scale=0.26]{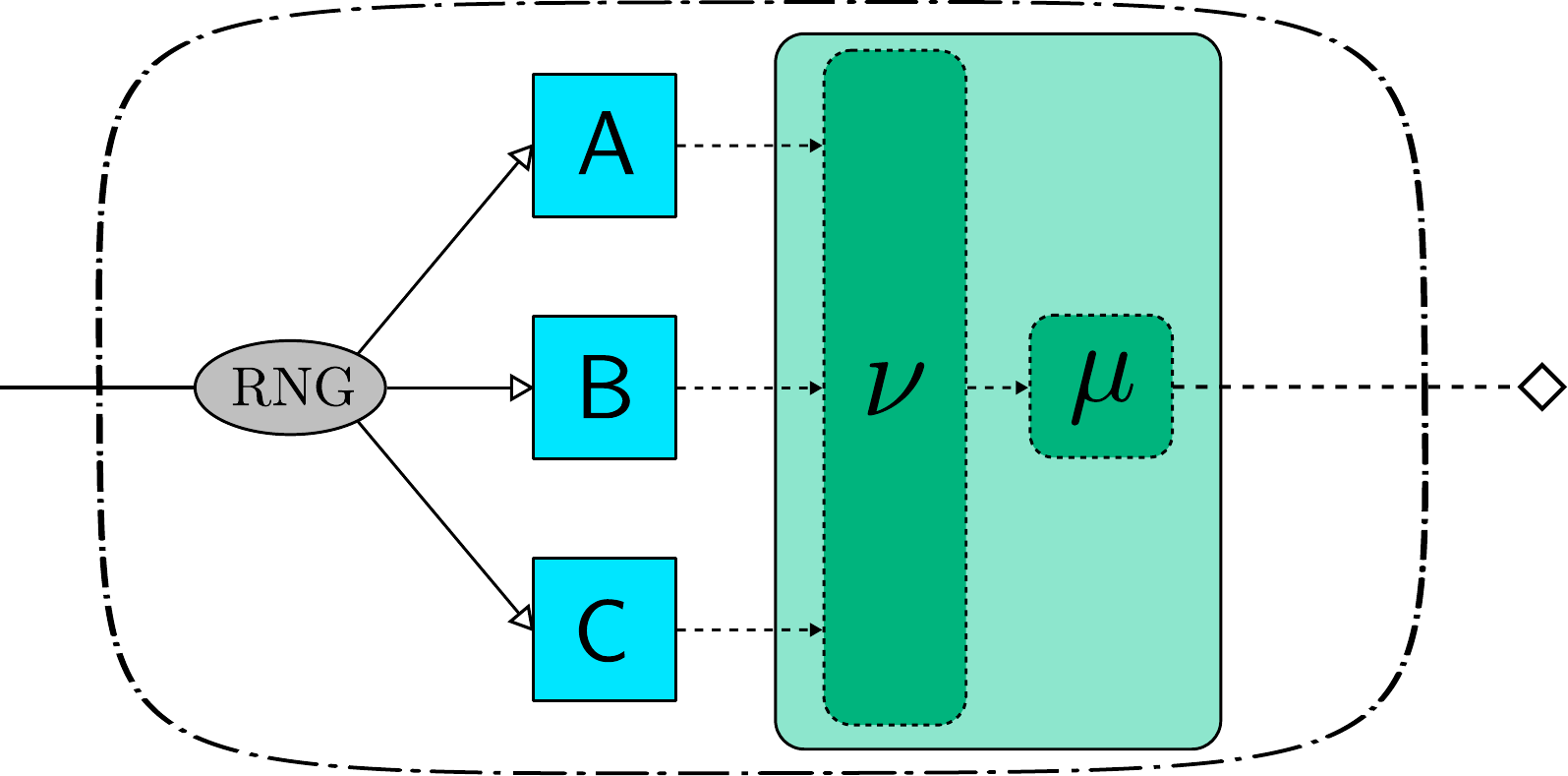}
\caption{\label{figure5} The set of simulable observables is
closed under postprocessings.}
\end{figure}

\subsection{Simulability and noise content} \label{sec:noise}

For an observable $\A$ and a set $\mathcal{N} \subset \mathcal{O}$ of noisy observables, we define \cite{FiHeLe17}
\begin{align*}
w(\A; \mathcal{N}) &= \sup \{ 0 \leq \lambda \leq 1 \ | \ \A = \lambda \N +(1- \lambda) \B  \\
& \quad \quad \quad \quad \quad \quad \quad \rm{\ for \ some\ } \N \in \mathcal{N} \rm{ \ and \ } \B \in \mathcal{O} \}
\end{align*}
as the \emph{noise content of $\A$ with respect to $\mathcal{N}$}. The noise content $w(\A; \mathcal{N})$ thus quantifies how much of $\A$ is in $\mathcal{N}$, which is taken to describe noise in the measurements. Contrary to external noise, i.e., noise that is added to the observables, the noise content gives us the amount of intrinsic noise that is already contained in the observable. The typical choice for the set of noisy observables is $\mathcal{N} = \trivial$, the set of trivial observables.

The noise content satisfies the following two properties \cite{FiHeLe17}:
\begin{itemize}
\item[a)] If $\noise$ is closed under postprocessings, then $w(\nu \circ \A; \noise) \geq w(\A;\noise)$ for all observables $\A$ and postprocessings $\nu$,

\item[b)] If $\noise$ is convex, then $w\left( \sum_i p_i \A^{(i)} ; \noise \right) \geq \sum_i p_i w(\A^{(i)}; \noise)$ for all mixtures of any set of observables $\{\A^{(i)}\}_i \subset \obs$.
\end{itemize}

We can now prove the intuitive result that we cannot simulate a less noisy observable from noisier ones:

\begin{proposition}
Let $\mathcal{B}$ be a set of simulators. If the set of noisy observables $\noise$ is closed under postprocessings and mixing, then any observable in $\simu{\mathcal{B}}$ has a noise content greater or equal than the smallest noise content of its simulating observables in $\mathcal{B}$.
\end{proposition}
\begin{proof}
Let $\A \in \simu{\mathcal{B}}$ so that
\begin{equation}
\A_x = \sum_{(i,x)} p_i \nu_{(i,y)x} \B^{(i)}_y = \sum_i p_i \left(\nu^{(i)} \circ \B^{(i)} \right)_x
\end{equation}
for some set of simulators $\{\B^{(i)}\}_i \subset \mathcal{B}$, probability distribution $(p_i)_i$ and postprocessing $\nu$. Now from properties a) and b) of the noise content it follows that
\begin{align*}
w(\A; \noise) &= w\left( \sum_i p_i \left(\nu^{(i)} \circ \B^{(i)} \right); \noise \right) \\
& \geq \sum_i p_i w\left( \nu^{(i)} \circ \B^{(i)}; \noise \right) \\
& \geq \sum_i p_i w\left( \B^{(i)} ; \noise \right) \\
& \geq \sum_i p_i \min_k w\left( \B^{(k)} ; \noise \right) \\
& = \min_k w\left( \B^{(k)} ; \noise \right).
\end{align*}

If there is an observable $\B \in \mathcal{B}$ such that $w(\B;\noise) \leq w(\B^{(i)}; \noise)$ for all $i$, then $w(\A; \noise) \geq w(\B;\noise)$.
\end{proof}
 We note that the set of trivial observables $\trivial$ is indeed convex and closed under postprocessings.

\subsection{Simulation irreducible observables}

Clearly, an observable $\A$ can be simulated by a subset
$\mathcal{B}$ whenever $\mathcal{B}$ contains $\A$, or more
generally, if there is $\B \in \mathcal{B}$ such that $\B$ is
postprocessing equivalent to $\A$. Those observables for which
this is the only way that they can be simulated we call simulation
irreducible:

\begin{definition}
An observable $\A$ is \emph{simulation irreducible} if for any subset $\mathcal{B} \subset \obs$, we have $\A \in \simu{\mathcal{B}}$ only if there is $\B \in \mathcal{B}$ such that $\A\leftrightarrow \B$.
\end{definition}

Simulation irreducibility thus means that the only way we can simulate such observable is essentially with the observable itself. We obtain a following characterization for the simulation irreducible observables.

\begin{proposition}\label{prop:sim-irred}
An observable is simulation irreducible if and only if it is postprocessing
clean and postprocessing equivalent to an extreme observable.
\end{proposition}

\begin{proof}
Let $\A$ be postprocessing clean and postprocessing equivalent to an extreme observable $\widetilde{\A}$ so that there exist postprocessings $\mu$ and $\eta$ such that $\widetilde{\A} = \mu \circ \A$ and $\A = \eta \circ \widetilde{\A}$. Suppose that $\A \in \simu{\mathcal{B}}$ for some set of simulators $\mathcal{B}$, i.e., there exists a probability distribution $(p_i)_i$ and postprocessings $\nu^{(i)}$ such that $\A_y = \sum_{i,x} p_i \nu^{(i)}_{xy} \B_x^{(i)} $ for some $\B^{(i)}$'s in $\mathcal{B}$.
We can assume that $p_i \neq 0$ for every $i$ as if this is not the case, we simply drop those terms away.
We can now write
\begin{align*}
\widetilde{\A}_z &= \sum_y \mu_{yz} \A_y = \sum_{i,x,y} p_i   \nu^{(i)}_{xy}\mu_{yz} \B_x^{(i)}\\
&= \sum_i p_i \sum_x \left( \sum_y  \nu^{(i)}_{xy} \mu_{yz} \right) \B^{(i)}_x \\
&= \sum_i p_i \sum_x \left( \mu \circ \nu^{(i)} \right)_{xz} \B^{(i)}_x \\
&= \sum_i p_i \left( \mu \circ \nu^{(i)} \circ \B^{(i)}\right)_z
\end{align*}
for all outcomes $z$. From the extremality of $\widetilde{\A}$ it follows that $\mu \circ \nu^{(i)} \circ \B^{(i)}=\widetilde{\A}$ for all $i$, and therefore $\A = \eta \circ \mu \circ \nu^{(i)} \circ \B^{(i)}$ for all $i$. Since $\A$ is postprocessing clean, this means that $\B^{(i)}\leftrightarrow \A$ for all $i$. Therefore, $\A$ is simulation irreducible.

Now let $\A$ be a simulation irreducible observable.
First, $\A$ has to be postprocessing clean; otherwise there exists an observable $\B$ such that $\B$ is not a postprocessing of $\A$ but $\A \in \simu{\B}$.
Secondly, if $\A$ is extreme, we are done, so let us consider the case when $\A$ is not extreme.
Then there exists some set of extreme observables $\mathcal{B}= \{\B^{(i)}\}_i$ such that $\A$ has a convex decomposition $\A = \sum_i \lambda_i \B^{(i)}$.
In particular, $\A \in \simu{\mathcal{B}}$ and since $\A$ is simulation irreducible, there exists some $k$ such that $\A\leftrightarrow \B^{(k)}$, where now $\B^{(k)}$ is extreme.
\end{proof}

We see that both postprocessing cleanness and postprocessing equivalence to an extreme observable are truly needed for simulation irreducibility.

\begin{example}
\emph{(postprocessing clean but not simulation irreducible quantum observable.)}
There are postprocessing clean quantum observables that are not simulation irreducible. For instance, the four-outcome qubit observable $\A$, related to the POVM $A(\pm 1)= \tfrac{1}{4} (\id \pm \sigma_x)$ and $A(\pm 2)= \tfrac{1}{4} (\id \pm \sigma_y)$, consists of linearly dependent but pairwisely linearly independent effects.
Therefore, $\A$ is not simulation irreducible even though it is postprocessing clean.
In fact, $\A$ can be obtained from two dichotomic observables $\X$ and $\Y$ as a mixture, where the corresponding POVMs are $X(\pm 1) = \tfrac{1}{2} ( \id \pm \sigma_x)$ and $Y(\pm 1) = \tfrac{1}{2} ( \id \pm \sigma_y)$, respectively.
\end{example}

We recall from the end of Sec. \ref{sec:p-p} that for each
observable $\A$, we can form an observable
$\hat{\A}\leftrightarrow\A$ such that the effects of $\hat{\A}$
are pairwisely linearly independent. By using the previous
propositions we find a more practical characterization of
simulation irreducibility.

\begin{corollary} \label{cor:irr-indec-extreme}
An observable $\A$ is simulation irreducible if and only if
$\hat{\A}$ is indecomposable and extreme, i.e., it consists of
linearly independent indecomposable effects.
\end{corollary}

\begin{proof}
Firstly, let $\A$ be simulation irreducible. By
Proposition~\ref{prop:sim-irred}, $\A$ is postprocessing clean and
postprocessing equivalent to an extreme observable $\B$. From
Proposition~\ref{prop:pp-clean}, we see that $\A$ is indecomposable from
which it follows that also the pairwise linearly independent
observable $\hat{\A}$ is indecomposable. What remains to show is
that the effects of $\hat{\A}$ are actually linearly independent.
Since $\B$ is extreme, by Proposition~\ref{prop:nonzero-extreme} the
nonzero effects of $\B$ are linearly independent. As $\hat{\B}$ is
formed by combining the pairwise linearly dependent effects of $\B$, we
have that $\hat{\B} = \B$ (without the possible zero effects of
$\B$). Thus, $\hat{\B}$ is extreme. Since $\hat{\B} = \B \leftrightarrow \A$
is pairwise linearly independent, we have by the uniqueness of
$\hat{\A}$ that $\hat{\B}$ is a bijective relabelling of
$\hat{\A}$. Hence, $\hat{\A}$ is extreme. Because $\hat{\A}$ is
also postprocessing clean, by Proposition~\ref{prop:p-p-clean-extreme}
it consists of linearly independent effects.

Second, suppose that $\hat{\A}$ consists of linearly independent
indecomposable effects. By Proposition~\ref{prop:pp-clean} $\hat{\A}$ is
postprocessing clean so that by taking into account that the
effects of $\hat{\A}$ are linearly independent we have by
Proposition~\ref{prop:p-p-clean-extreme} that $\hat{\A}$ is extreme.
Since also $\A$ is postprocessing clean and
$\hat{\A}\leftrightarrow\A$ is extreme, from
Proposition~\ref{prop:sim-irred} we conclude that $\A$ is simulation
irreducible.
\end{proof}

We would expect that an observable that is not simulation irreducible is reducible in the sense that it can be simulated by some simulation irreducible observables. Indeed, we can show that this is the case and even holds with a finite number of simulators \cite{HaHePe12}.

\begin{proposition}\label{prop:finite}
For every observable $\A$, there is a finite collection $\mathcal{B}^\A$ of simulation irreducible observables such that $\A \in \simu{\mathcal{B}^\A}$.
\end{proposition}

\begin{proof}
Let $\A$ be an observable with an outcome set $\Omega$. Each
effect $\A_x$ can be decomposed into indecomposable effects
$a^{(x)}_i$ such that $\A_x = \sum_{i=1}^{r_x} a^{(x)}_i$ for some
finite $r_x$. As in the proof of Proposition~\ref{prop:pp-clean}, we
denote $r = \max_{x \in \Omega} r_x$ and define an indecomposable
observable $\B$ with an outcome set $\{1,\ldots,r\} \times \Omega$
by $\B_{(i,x)} = a^{(x)}_i$ if $i \leq r_x$ and $\B_{(i,x)}=o$
otherwise. Let us consider the pairwise linearly independent
observable $\hat{\B}$. Observable $\A$ is then a postprocessing
of $\hat{\B}$ (as it is of $\B$).

If $\hat{\B}$ is extreme, we are done. Otherwise, $\hat{\B}$ is
not extreme so its effects are linearly dependent,
i.e., there exist numbers $\beta_i \in \real$ such that $ \sum_{i}
\beta_i \hat{\B}_i=o$ with $\sum_i | \beta_i | >0$. Note that we
must have both positive and negative $\beta_i$'s.

We denote $\kappa_+ = \max_i \beta_i > 0$ and $\kappa_- = \min_ i \beta_i <0$ and consider two observables $\C$ and $\D$ defined as follows:
\begin{eqnarray}
  \C_i = (1 - \beta_i/\kappa_+) \hat{\B}_i  \\
  \D_i = (1 - \beta_i/\kappa_-) \hat{\B}_i
\end{eqnarray}
for all outcomes $i$. We note that both $\C$ and $\D$ have one nonzero outcome less than $\hat{\B}$ since for some indices $j$ and $k$ we have that $\beta_j = \kappa_+$ and $\beta_k = \kappa_-$ so that $\C_j = \D_k =o$. Since the effects of $\hat{\B}_i$ are indecomposable, the observables $\C$ and $\D$ are also indecomposable. By setting $\lambda = \kappa_+/(\kappa_+-\kappa_-) $ we find that
\begin{equation}
\hat{\B}_i = \lambda \C_i + (1-\lambda)\D_i
\end{equation}
for all $i$.

Thus, $\hat{\B}$ can be expressed as a mixture of two indecomposable observables with one less nonzero outcome. If the nonzero effects of $\C$ and $\D$ are still linearly dependent we continue this procedure until we eventually have reduced the outcomes with finite steps in such a way that the resulting observables, denoted by the set $\mathcal{B}^{\A}$, have linearly independent effects. Since the indecomposability is preserved over the procedure, the observables in $\mathcal{B}^{\A}$ are simulation irreducible.
\end{proof}

\begin{example}
\emph{(Simulation irreducible quantum observables.)} As explained
before, a quantum observable $\A$ is postprocessing clean if and
only if each operator $A(x)$ is rank-1. To check if such an
observable is simulation irreducible, we can construct a minimally
sufficient representative $\hat{\A}$ of the postprocessing
equivalence class of $\A$ as explained in Sec. \ref{sec:p-p}
and then check the linear independence of the effects of
$\hat{\A}$. In $d$-dimensional quantum theory $\quant_d$, the
maximal number of linearly independent operators is $d^2$. For any
integer $d,\ldots,d^2$, one can construct an extreme
postprocessing clean observable \cite{HaHePe12}. Further, two
POVMs $A$ and $B$ consisting of rank-1 operators are seen to be
postprocessing equivalent if and only if the set of ranges
$\cup_x \{ ran(A(x))\}$ and $\cup_y \{ ran(B(y))\}$ are the same.
There is therefore a continuum of postprocessing inequivalent
simulation irreducible observables in $\quant_d$ for any $d\geq
2$.
\end{example}

\section{Limitation on the number of observables}

\subsection{Minimal simulation number}

A set of observables $\mathcal{A}$ is \emph{compatible} if there
exists an observable $\G$ such that every observable in
$\mathcal{A}$ can be post-processed from $\G$. Thus, if
$\mathcal{A}= \{\A^{(1)}, \ldots,\A^{(m)}\}$ is a collection of
$m$ observables with outcome sets $X_1, \ldots, X_m$, then
$\mathcal{A}$ is compatible if there exists an observable $\G$
with an outcome set $Y$ and postprocessings $\nu^{(i)}: Y \to X_i$, $i=1, \ldots,m$, such that
\begin{equation}
\A^{(i)} = \nu^{(i)} \circ \G
\end{equation}
for all $i=1, \ldots,m$. This means that by measuring only $\G$  we can implement a measurement of any observable in $\mathcal{A}$ just by choosing a suitable postprocessing.

As explained in Ref. \cite{GuBaCuAc17}, simulability can be seen as an
extension of compatibility. In fact, if we consider an observable
$\G$ and the set of $\G$-simulable observables $\simu{\G}$, we see
that every simulation in $\simu{\G}$ comprises mixing a
single observable $\G$ so that by reducing the multiplicity the
mixing becomes trivial. Then $\simu{\G}$ is seen to be just the
set of postprocessings of $\G$, and so $\simu{\G}$ is a
compatible set of observables and every subset $\mathcal{A}
\subseteq \simu{\G}$ is compatible. On the other hand, if there is
a compatible set $\mathcal{A}$ such that every observable can be
post-processed from $\G$, then clearly $\mathcal{A} \subseteq
\simu{\G}$.

If a subset $\mathcal{A}$ is not compatible, then there is no
single observable $\G$ such that $\mathcal{A} \subseteq
\simu{\G}$. But we can still search for the minimal collection of
simulators that can produce $\mathcal{A}$. This leads to the
following definition.

\begin{definition}
For a subset $\mathcal{A}\subseteq\obs$, we denote by
$\smin{\mathcal{A}}$ the minimal number of observables
$\B^{(1)},\ldots,\B^{(n)}$, if they exist, such that
$\mathcal{A}\subseteq\simu{\B^{(1)},\ldots,\B^{(n)}}$. Otherwise
we denote $\smin{\mathcal{A}}=\infty$. We call
$\smin{\mathcal{A}}$ \emph{the minimal simulation number for
$\mathcal{A}$}.
\end{definition}

Let us consider a finite set
$\mathcal{A}=\{\A^{(1)},\ldots,\A^{(m)}\}\subset\obs$. Clearly,
$\smin{\A^{(1)},\ldots,\A^{(m)}}\leq m$. Further, if $k$
observables among $\A^{(1)},\ldots,\A^{(m)}$ are compatible, then
$\smin{\A^{(1)},\ldots,\A^{(m)}}\leq m-k+1$. This indicates that
the hypergraph structure of the compatibility relation of the set
$\mathcal{A}$, as defined in Ref. \cite{KuHeFr14}, relates to
$\smin{\mathcal{A}}$; by identifying the largest subset of
compatible observables we get an upper bound for
$\smin{\mathcal{A}}$. This connection is, however, only in one
direction, as observed in Ref. \cite{GuBaCuAc17}. Namely, there exists a
set $\{\A,\B,\C\}$ of three quantum observables such that no pair
is compatible, but still $\smin{\A,\B,\C}=2$. The following
example is slightly different from Example 1 in Ref. \cite{GuBaCuAc17},
which consisted of four observables.

\begin{example}
\emph{(There exist three pairwisely incompatible quantum
observables $\A,\B,\C$ such that $\smin{\A,\B,\C}=2$.)} We denote
$A(\pm)= \half (\id \pm \sigma_x)$, $B(\pm) = \half (\id \pm
\sigma_y)$, and $C_t(\pm) = \half (\id \pm t
(\sigma_x+\sigma_y)/\sqrt{2})$, where $0<t<1$ is a parameter to be
specified. Since $\A$ (resp. $\B$) consists of projections, any
observable compatible with it must commute with it (see, e.g.,
Ref. \cite{HeReSt08}). Hence, $\C_t$ is incompatible with both
for any $0 < t \leq 1$. We clearly have $\A,\B \in \simu{\A,\B}$,
and we also have $\C_t \in \simu{\A,\B}$ whenever $t \leq
1/\sqrt{2}$. Namely, by taking the equal mixture of $A$ and $B$ we
get a POVM $C_{1/2}(\pm)=\half (\id \pm (\sigma_x+\sigma_y)/2)$.
By using a postprocessing matrix
$$
\half \left(\begin{array}{cc}1+\sqrt{2} t & 1-\sqrt{2} t \\1-\sqrt{2} t & 1+\sqrt{2} t\end{array}\right)
$$
we get $C_t$ from $C_{1/2}$ for any $t \leq 1/\sqrt{2}$.
(The fact $\C\notin\simu{\A,\B}$ for $t>1/\sqrt{2}$ will be shown in Example \ref{ex:qubit}.)
\end{example}

\subsection{Connection to $k$-compatibility}

A joint measurement of observables $\A^{(1)},\ldots,\A^{(n)}$ means that we can simultaneously implement their measurements using a single observable, even if only one input system is available.
In the context of quantum observables, this notion has been recently generalized to the case where it is assumed that we have access to $k$ copies \cite{CaHeReScTo16}.
We can then make a collective measurement on a state $s^{\otimes k}$.
After obtaining a measurement outcome, we can make copies of the outcome and post-process each copy in a preferred way.
This leads to the following notion:
 Observables $\A^{(1)},\ldots,\A^{(n)}$ on sets $\Omega_1,\ldots,\Omega_n$, respectively, are \emph{$k$-compatible} if there exists an observable $\G$ with an outcome set $\Omega_0$ acting on the state space $\state^{\otimes k}$,  and stochastic matrices $\nu_1,\ldots,\nu_n$ with $\nu_i : \Omega_i\times\Omega_0\to [0,1]$, such that
\begin{equation}\label{eq:post}
\sum_{y\in\Omega_0}  \nu_i(x_i,y) \G_y(s^{\otimes k}) = \A^{(i)}_{x_i}(s)
\end{equation}
for all $i=1,\ldots,n$,  $x_i\in\Omega_i$, and $s\in\state$.
This definition obviously requires that we have specified the tensor product of two state spaces.

As in the usual case of compatibility, we can restrict to a
special kind of observables and postprocessings when deciding
whether a collection of observables is $k$-compatible.
Namely, suppose that $\Omega$ is the Cartesian product $\Omega=
\Omega_1 \times \cdots \times \Omega_n$ and that $\C$ is an
observable with this outcome set. One particular type of
postprocessing comes from ignoring all but the $i$th component $x_i$ of
a measurement outcome $(x_1,\ldots,x_n)$. This kind of
postprocessing gives the $i$th marginal of $\C$, which we denote
as $\C^{[i]}$, i.e.,
\begin{equation}
\C^{[i]}_{x} = \sum_{\ell \neq i}\sum_{x_\ell}\C_{x_1,\ldots,x_{i-1},x,x_{i+1}\ldots,x_n} \, .
\end{equation}
Suppose there exits an observable $\G$ and stochastic matrices $\nu_1,\ldots,\nu_n$ such that \eqref{eq:post} holds.
We define $\C$ as
\begin{equation}
\C_{x_1,\ldots,x_n} = \sum_{y\in\Omega_0} \nu_1(x_1,y) \ldots  \nu_n(x_n,y) \G_y \, ,
\end{equation}
in which case $\C$ is an observable with the outcome set $\Omega_1\times\ldots\times\Omega_n$ and $\C^{[i]} = \A^{(i)}$.

\begin{proposition}\label{prop:k-comp}
If observables $\A^{(1)},\ldots,\A^{(n)}$ can be simulated by $k$ observables (i.e., $\smin{\A^{(1)},\ldots,\A^{(n)}} \leq k$), then they are $k$-compatible.
\end{proposition}

\begin{proof}
Let $\Lambda_i$ denote the outcome set of an observable $\A^{(i)}$
for all $i=1, \ldots, n$ and let $\B^{(1)},\ldots,\B^{(k)}$ be observables with an outcome set $\Omega$ such that
$\{ \A^{(1)},\ldots,\A^{(n)} \} \subseteq \simu{\{
\B^{(1)},\ldots,\B^{(k)}\}}$. Thus, there exist $n$ probability
distributions $(p^{(i)}_j)_{j=1}^k$, $i=1, \ldots,n$ and $n$
postprocessings $\nu^{(i)}: \{1,\ldots,k\} \times \Omega \to
\Lambda_i$, $i=1, \ldots,n$ such that
\begin{equation}
\A^{(i)}_{y_i} = \sum_{j=1}^k \sum_{x \in \Omega} p^{(i)}_j \nu^{(i)}_{(j,x)y_i} \B^{(j)}_x
\end{equation}
for all $y_i \in \Lambda_i$ and all $i=1, \ldots,n$.

We define an observable $\G$ with an outcome set $\Omega^k$ on $\state^{\otimes k}$ as
\begin{align}
\G_{x_1,\ldots,x_k} = \B^{(1)}_{x_1} \otimes \cdots \otimes \B^{(k)}_{x_k}
\end{align}
for all $(x_1,\ldots,x_k)\in \Omega^k$, and postprocessings $\mu^{(i)}: \Omega^k \to \Lambda_i$ for all $i=1, \ldots,n$ by
\begin{equation}
\mu^{(i)}_{\vec{x} y_i} = \sum_{j=1}^k p^{(i)}_j \nu^{(i)}_{(j,x_j)y_i}
\end{equation}
for all $\vec{x}=(x_1,\ldots,x_k)\in \Omega^n$ and $y_i \in \Lambda_i$. We now see that
\begin{align*}
\sum_{\vec{x} \in \Omega^k} \mu^{(i)}_{\vec{x}y_i} \G_{\vec{x}}(s^{\otimes k}) &= \sum_{x_1 \in \Omega} \cdots \sum_{x_k \in \Omega} \mu^{(i)}_{(x_1, \ldots,x_k)y_i} \prod_{l=1}^k \B^{(l)}_{x_l}(s) \\
&= \sum_{j=1}^k \sum_{x_1 \in \Omega} \cdots \sum_{x_k \in \Omega} p^{(i)}_j \nu^{(i)}_{(j,x_j)y_i} \prod_{l=1}^k \B^{(l)}_{x_l}(s) \\
&= \sum_{j=1}^k \sum_{x_j \in \Omega} p^{(i)}_j \nu^{(i)}_{(j,x_j)y_i} \B^{(j)}_{x_j}(s) \\
&= \A^{(i)}_{y_i}(s)
\end{align*}
for all states $s \in \mathcal{S}$, outcomes $y_i \in \Lambda_i$
and $i=1,\ldots,n$. Hence, the observables
$\A^{(1)},\ldots,\A^{(n-1)}$, and $\A^{(n)}$ are $k$-compatible.
\end{proof}

\begin{example}
\emph{(Triplet of orthogonal qubit observables.)} We denote
$X_t(\pm)= \half (\id \pm t\sigma_x)$, $Y_t(\pm) = \half (\id \pm
t\sigma_y)$, and $Z_t(\pm) = \half (\id \pm t\sigma_z)$, where
$t\in [0,1]$ is a noise parameter. For $t=1$ these observables are
simulation irreducible and therefore $\smin{\X_1,\Y_1,\Z_1}=3$. The
triplet is compatible if and only if $0\leq t \leq 1/\sqrt{3}$, so
for exactly those values $\smin{\X_t,\Y_t,\Z_t}=1$. It was proved in Ref. 
\cite{CaHeReScTo16} that this triplet is $2$-compatible if and
only if $0 \leq t \leq \sqrt{3}/2$;, hence we conclude that
$\smin{\X_t,\Y_t,\Z_t}=2$ for $1/\sqrt{3} < t \leq \sqrt{3}/2$. From
these results, we cannot conclude the minimal simulation number for
values $\sqrt{3}/2<t<1$.
\end{example}

\section{Limitation on the number of outcomes}

\subsection{Effective number of outcomes}

We denote by $\obs_n$ the set of observables with the outcome set $\{1,\ldots,n\}$.

\begin{definition}
An observable $\A$ \emph{has effectively $n$ outcomes} if $n$ is the least number such that $\A$
can be simulated by $\obs_n$. We denote by $\eff_{n}$ the set of
all those observables that have effectively $n$ or less outcomes.
Further, we say that a subset $\obs'\subseteq \obs$ is
\emph{effectively $n$-tomic} if $\obs'\subseteq \eff_n$.
\end{definition}

Clearly, if an observable $\A$ can be simulated by $\obs_n$, then also any postprocessing of $\A$ can be simulated by $\obs_n$.
Further, a mixture of two $n$-tomic observables is at most $n$-tomic.
Therefore, the sets $\eff_{1} \subseteq \eff_{2} \subseteq \cdots $ are convex and closed under postprocessing.
The set $\eff_{1}$ consists exactly of all trivial observables, i.e., observables of the form $\T(x) = t(x) u$.

As explained at the end of Sec. \ref{sec:p-p}, for each
observable $\A$ we can form an observable
$\hat{\A}\leftrightarrow\A$ such that the effects of $\hat{\A}$
are pairwisely linearly independent. It follows from the
construction of $\hat{\A}$ that the number of outcomes of
$\hat{\A}$ is at most the number of outcomes of $\A$. If $\A$ is
simulation irreducible, then the effective number of outcomes of
$\A$ is equal to the number of outcomes of $\hat{\A}$.

By Proposition~\ref{prop:finite}, every observable can be simulated with
simulation irreducible observables. Therefore, the maximal
effective number of outcomes in a given theory can be concluded by
looking at the extreme simulation irreducible observables. For
instance, for a set of quantum observables $\quant_d$ in a
$d$-dimensional quantum theory, the maximal effective number of
outcomes is $d^2$. We will calculate the maximal effective number
of outcomes for some other states spaces in Sec.
\ref{sec:non-quantum}.

\begin{example}(\emph{Informationally complete quantum observables.})
An observable $\A$ is called \emph{informationally complete} if $\A(s_1)\neq \A(s_2)$ for any two states $s_1 \neq s_2$.
A quantum observable $\A$ is informationally compelete if and only if the respective set of POVM elements $\{ A(x): x \in \Omega \}$ spans the vector space $\mathcal{L}_s(\hi)$ of all self-adjoint operators \cite{Busch91}.
It follows that an informationally complete observable on $\quant_d$ has at least $d^2$ outcomes.
However, it is easy to construct an informationally complete observable which is effectively dichotomic.
For this purpose, fix linearly independent operators $B_1,\ldots,B_{d^2}\in\mathcal{L}_s(\hi)$.
For each $j$, we define a dichotomic POVM $A^{(j)}$ as
\begin{equation*}
A^{(j)}(\pm ) = \half (\id \pm B_j/\no{B_j}) \, .
\end{equation*}
The equal mixture of these POVMs is then
\begin{equation*}
A(\pm ,j) = \tfrac{1}{2d^2} (\id \pm B_j/\no{B_j}) \, ,
\end{equation*}
and then the span of the elements of $A$ is clearly $\mathcal{L}_s(\hi)$.
Therefore, the corresponding observable $\A$ is informationally complete but effectively dichotomic.

The mathematical criterion for an observable to be informationally complete is the same in every general probabilistic theory \cite{SiSt92}, and one can show that the previous conclusion is valid in any general probabilistic theory: There exists an informationally complete observable which is effectively dichotomic.
\end{example}

\subsection{Dichotomic observables}

As a particular example, we will take a closer look at dichotomic and effectively dichotomic observables. We will see that in many cases they have a simple geometrical characterization.

\begin{proposition}\label{prop:conv-sim}
Let $\mathcal{B}=\{\B^{(i)}\}_{i=1}^m$ be a collection of $m$
observables with an outcome set $\Omega$. For a dichotomic
observable $\A$ with effects $\A_+$ and $\A_-$ the following
implication holds:
$$
\A_+ \in \mathrm{conv}\left(\{\{\B^{(i)}_x\}_{i,x},o,u\}\right) \quad \Rightarrow \quad \A \in \simu{\mathcal{B}}.
$$
\end{proposition}
\begin{proof}
Let $\A_+ \in \mathrm{conv}\left(\{\{\B^{(i)}_x\}_{i,x},o,u\}\right)$ so that
\begin{equation}\label{eq:conv-effects}
\A_+ = \sum_{i=1}^m \sum_{x \in \Omega} \eta_{ix} \B^{(i)}_x + \lambda u + \mu o
\end{equation}
for some positive numbers $\eta_{ix}, \lambda,\mu \in \real$ for all $i=1,\ldots,m$ and $x \in \Omega$ such that $\sum_{i,x} \eta_{ix} + \lambda+\mu =1$. From the normalization of the observables in $\mathcal{B}$, it follows that for any probability distribution $(q_i)_{i=1}^m$ we have
\begin{equation}
u= \sum_{i,x} q_i \B^{(i)}_x.
\end{equation}

By plugging the previous expression in Eq. \eqref{eq:conv-effects} and neglecting the term with the zero effect $o$, we have that
\begin{equation}
\A_+ = \sum_{i,x} (\eta_{ix}+\lambda q_i) \B^{(i)}_x = \sum_{i,x} \tilde{\eta}_{ix} \B^{(i)}_x,
\end{equation}
where we have denoted $\tilde{\eta}_{ix}= \eta_{ix}+\lambda q_i$ for all $i=1, \ldots,m$ and $x \in \Omega$. We can now introduce a probability distribution $(p_i)_{i=1}^m$ by
\begin{align*}
& p_i = \max_{x\in \Omega} \tilde{\eta}_{ix}, \quad  i=1, \ldots,m-1,\\
& p_m = 1- \sum_{i=1}^{m-1} p_i \, .
\end{align*}
It is straightforward to check that $(p_i)_i$ actually forms a probability distribution.

We define a postprocessing $\nu: \{1,\ldots,m\} \times \Omega \to \{+,-\}$ by
\begin{align}
& \nu_{(i,x)+}=
\begin{cases}
\dfrac{\tilde{\eta}_{ix}}{p_i} & {\rm if \ } p_i \neq 0, \\
0, & {\rm if \ } p_i= 0,
\end{cases} \\
& \nu_{(i,x)-} = 1- \nu_{(i,x)+}
\end{align}
for all $i=1, \ldots,m$ and $x \in \Omega$. We see that indeed $\nu_{(i,x)\pm} \in [0,1]$ and $\nu_{(i,x)+} + \nu_{(i,x)-} =1$ for all $i=1, \ldots,m$ and $x \in \Omega$, so $\nu$ is a legitimate postprocessing. Hence, there exists a probability distribution $(p_i)_i$ and a postprocessing $\nu$ such that
\begin{equation}
\A_{\pm} = \sum_{i,x} \nu_{(i,x)\pm} p_i \B^{(i)}_x
\end{equation}
so that $\A \in \simu{\mathcal{B}}$.
\end{proof}

The previous proposition only considers simulated observables
which have only two outcomes. We see that the proposition can in
fact be extended to cover simulated observables with more outcomes
at the expense of the form of the simulator observables.

\begin{proposition}\label{prop:conv-sim-lin-ind}
Let $\mathcal{B}=\{\B^{(i)}\}_{i=1}^m$ be a collection of $m$ dichotomic observables such that the set $\left\{ u, \{\B^{(i)}_+\}_{i=1}^m \right\}$ is linearly independent. For an observable $\A$ with an outcome set $\Lambda$,  the following implication holds:
$$
\A_y \in \mathrm{conv}\left(\{\{\B^{(i)}_\pm\}_{i},o,u\}\right) \ \forall y \in \Lambda  \quad \Rightarrow \quad \A \in \simu{\mathcal{B}}.
$$
\end{proposition}
\begin{proof}
Let $\A$ be an observable with outcome set $\Lambda$ such that $\A_y \in \mathrm{conv}\left(\{\{\B^{(i)}_\pm\}_{i},o,u\}\right)$ for all $y \in \Lambda$ so that
\begin{eqnarray}\label{eq:lin-ind}
\A_y &=& \sum_i \left( \lambda^{(i,y)}_+ \B^{(i)}_+ + \lambda^{(i,y)}_- \B^{(i)}_- \right) + \lambda^{(u,y)} u + \lambda^{(o,y)} o \nonumber \\
&=& \sum_i \left( \omega^{(i,y)}_+ \B^{(i)}_+ + \omega^{(i,y)}_- \B^{(i)}_-\right),
\end{eqnarray}
where $\{\{\lambda^{(i,y)}_\pm\}_i , \lambda^{(u,y)}, \lambda^{(o,y)}\}$ is a probability distribution for all $y \in \Lambda$ and $\omega^{(i,y)}_\pm = \lambda^{(i,y)}_\pm + \frac{1}{m} \lambda^{(u,y)}$ for all $i= 1, \ldots,m$ and $y \in \Lambda$. Here we have taken into account that $u = \frac{1}{m} \sum_i(\B^{(i)}_+ + \B^{(i)}_-)$.

Because of the normalization of $\A$ we have that
\begin{eqnarray}
u &=& \sum_y \A_y = \sum_i \left( \omega^{(i)}_+ \B^{(i)}_+ +
\omega^{(i)}_- \B^{(i)}_- \right)
\nonumber\\
&=&  \left(\sum_i  \omega^{(i)}_- \right) u + \sum_i  \left( \omega^{(i)}_+ -\omega^{(i)}_- \right) \B^{(i)}_+,
\end{eqnarray}
where $\omega^{(i)}_\pm = \sum_y \omega^{(i,y)}_\pm \geq 0$ for
all $i=1, \ldots, m$.

Since effects $u$, $\B^{(1)}_+,\ldots, \B^{(m)}_+$ are linearly independent, we conclude that
$\omega^{(i)}_+ = \omega^{(i)}_- =: p_i$ for all $i= 1, \ldots,m$
and $\sum_i p_i = 1$. We can then define a postprocessing $\nu:
\{1, \ldots,m\}\times \{+,-\} \to \Lambda$ by setting
\begin{equation}
\nu_{(i,\pm)y} =
\begin{cases}
\dfrac{\omega^{(i,y)}_\pm}{p_i}, & {\rm if \ } p_i \neq 0,\\
 \dfrac{1}{m},& {\rm if \ } p_i=0.
\end{cases}
\end{equation}
From Eq. \eqref{eq:lin-ind}, we can now confirm that
\begin{equation}
\A_y = \sum_i p_i \left( \nu_{(i,+)y} \B^{(i)}_+ + \nu_{(i,-)y} \B^{(i)}_- \right)
\end{equation}
for all $y \in \Lambda$ so that $\A \in \simu{\mathcal{B}}$.
\end{proof}

We note that if there is only one dichotomic simulator observable $\B$, then $\{u, \B_+\}$ is linearly independent if and only if $\B$ is nontrivial. In the case of one simulator $\B$ we can even have more outcomes for $\B$ provided that the effects of $\B$ are linearly independent.

\begin{proposition}\label{prop:conv-sim-single-lin-ind}
Let $\B$ be an observable with linearly independent effects and an outcome set $\Omega$. For an observable $\A$ with an outcome set $\Lambda$, the following implication holds:
$$
\A_y \in \mathrm{conv}\left(\{\{\B_x\}_{x},o,u\}\right) \quad \forall y \in \Lambda \quad \Rightarrow \quad \A \in \simu{\B}.
$$
\end{proposition}
\begin{proof}
Each effect $\A_y$ can be expressed as a convex decomposition into
the effects $\B_x$, $o$, and $u$ so that
\begin{equation}
\A_y = \sum_x \lambda^{(y)}_x \B_x + \lambda^{(y)}_o o + \lambda^{(y)}_u u
\end{equation}
for all $y \in \Lambda$ for some positive numbers $\lambda^{(y)}_x, \lambda^{(y)}_o$ and $\lambda^{(y)}_u$ such that $\sum_x \lambda^{(y)}_x +\lambda^{(y)}_o + \lambda^{(y)}_u =1$ for all $y \in \Lambda$. Since $u = \sum_x \B_x$, we have that
\begin{equation}
\A_y = \sum_x \left( \lambda^{(y)}_x + \lambda^{(y)}_u \right) \B_x
\end{equation}
for all $y \in \Lambda$. From the normalization of observables $\A$ and $\B$ it follows that
\begin{equation}
\sum_x \B_x = u = \sum_y \A_y = \sum_x \left[ \sum_y \left( \lambda^{(y)}_x + \lambda^{(y)}_u \right) \right] \B_x
\end{equation}
The linear independence of the effects $\B_x$ leads us to conclude that
$ \sum_y \left( \lambda^{(y)}_x + \lambda^{(y)}_u \right)=1$ for all $x \in \Omega$. Thus, if we define a mapping $\nu: \Omega \to \Lambda$ by
\begin{equation}
\nu_{xy} =  \lambda^{(y)}_x + \lambda^{(y)}_u
\end{equation}
for all $ x\in \Omega$ and $y \in \Lambda$, we see that now $\nu$ is a postprocessing and $\A = \nu \circ \B$.
\end{proof}

As an example of this, a simulation irreducible observable (or its minimally sufficient version) consists of linearly independent effects, and so in this case we have a sufficient condition for an observable to be simulated by it. However, we see that the
condition is not a necessary one and also that if we try to
increase the number of simulators then the proposition no longer
holds. The converse of Proposition~\ref{prop:conv-sim} is also seen to
be false in general.

\begin{example}[\it Simulation irreducible qubit observable]
\label{ex:qubit-sim-irr}
Let us consider a 4-outcome qubit observable $\B$ with effects
\begin{equation*}
B(i) = \dfrac{1}{4} \left(  \id + \vec{b}_i \cdot \vec{\sigma} \right), \quad i=1,2,3,4,
\end{equation*}
where
\begin{eqnarray}
\vec{b}_1 &=& \left( \dfrac{2 \sqrt{2}}{3}, 0, -\dfrac{1}{3} \right), \quad \vec{b}_2 = \left( -\dfrac{ \sqrt{2}}{3}, \sqrt{\dfrac{ 2}{3}}, -\dfrac{1}{3} \right), \nonumber \\
\vec{b}_3 &=& \left( -\dfrac{ \sqrt{2}}{3}, -\sqrt{\dfrac{ 2}{3}}, -\dfrac{1}{3} \right), \quad \vec{b}_4 = (0,0,1), \label{eq:qubit-tetra}
\end{eqnarray}
so that the four vectors form the vertices of a tetrahedron inside
a unit ball. Clearly, any set of three of the four vectors form a
linearly independent set and in fact the set of all four effects
is linearly independent. Furthermore, since the effects $B(i)$ are
rank-1 for all $i=1,2,3,4$, we have that $\B$ is
simulation irreducible.

To see that the converses of Proposition~\ref{prop:conv-sim} and
\ref{prop:conv-sim-single-lin-ind} do not hold, we define a
dichotomic qubit observable $\A$ by setting
\begin{eqnarray*}
A(+) &=& B(1) + B(2) = \dfrac{1}{2} \left[ \id + \left( \dfrac{\vec{b}_1 + \vec{b}_2}{2} \right) \cdot \vec{\sigma} \right], \\
A(-) &=& B(3) + B(4) = \dfrac{1}{2} \left[ \id + \left( \dfrac{\vec{b}_3 + \vec{b}_4}{2} \right) \cdot \vec{\sigma} \right].
\end{eqnarray*}

Clearly, $\A \in \simu{\B}$. We will show by contradiction that
$A(+)$ does not belong to the convex set of effects $O$, $\id$,
$B(1)$, $B(2)$, $B(3)$, $B(4)$. Suppose that $A(+) \in
\mathrm{conv}\left(\{B(1),B(2),B(3),B(4),O,\id\}\right)$, i.e.,
for $A(+)$ there exists a convex decomposition
\begin{align*}
A(+) &= \sum_{i=1}^4 \lambda_i B(i) + \lambda_5 \id + \lambda_6 O \\
&= \dfrac{1}{2} \left[ \left(\dfrac{\sum_{i=1}^4 \lambda_i}{2}+ 2 \lambda_5 \right) \id + \left( \dfrac{\sum_{i=1}^4 \lambda_i \vec{b}_i}{2} \right) \cdot \vec{\sigma} \right].
\end{align*}

By comparing the coefficients of $\id$ and the Pauli matrices, we arrive at the following two equations:
\begin{equation*}
\begin{cases}
\sum\limits_{i=1}^4 \lambda_i + 4 \lambda_5 = 2, \\
\sum\limits_{i=1}^4 \lambda_i \vec{b}_i = \vec{b}_1 + \vec{b}_2.
\end{cases}
\end{equation*}
By using the latter equation and the condition that $\sum_{i=1}^4 \vec{b}_i = \vec{0}$, we have that
\begin{equation}
(1+\lambda_4-\lambda_1) \vec{b}_1 +(1+\lambda_4-\lambda_2) \vec{b}_2 +(\lambda_4-\lambda_3) \vec{b}_3 = \vec{0}.
\end{equation}
Now the set $\{\vec{b}_1, \vec{b}_2, \vec{b}_3\}$ is linearly independent so that
\begin{equation}
\lambda_1 = \lambda_2 = 1+ \lambda_4 \quad \Rightarrow \quad \lambda_1 = \lambda_2 =1,
\end{equation}
which contradicts the fact that $\sum_{i=1}^6 \lambda_i =1$. Thus,
$A(+)$ cannot be contained in the convex hull of $O$, $\id$, and
the effects of $\B$. By similar arguments we see that since, for
example, the set  $\{\vec{b}_2, \vec{b}_3, \vec{b}_4\}$ is linearly
independent, then $A(-)$ also cannot be contained in the convex
hull of $O$, $\id$, and the effects of $\B$.

We also see that in Proposition~\ref{prop:conv-sim-lin-ind} both the
dichotomicity of observables $\B^{(i)}$ and the linear
independence of the effects $\{u, \{\B^{(i)}_+ \}_i \}$ is truly
needed: Define dichotomic observables $\C^{(i)}$ by setting
$C^{(i)}(+) = B(i)$ and $C^{(i)}(-) = \id -B_i$ for all
$i=1,2,3,4$. Clearly $B_i \in \mathrm{conv}\left( \{\{C^{(j)}(\pm)
\}_{j=1}^4, O, \id \} \right)$ for all $i=1,2,3,4$ and even the
effects $C^{(1)}(+),C^{(2)}(+),C^{(3)}(+)$ and $C^{(4)}(+)$ are
linearly independent by themselves but not with the unit effect
$\id$. If $\B \in \simu{\{\C^{(i)}\}_{i=1}^4}$,
then by the simulation irreducibility of $\B$ we would have that
$\B \leftrightarrow \C^{(k)}$ for some $k \in \{1,2,3,4\}$ which
is clearly not the case since when measuring $\C^{(k)}$ we only
get information about the outcome $k$ of the observable $\B$ and
not the other outcomes. This also happens when we define two
trichotomic observables $\D^{(1)}$ and $\D^{(2)}$ by setting
\begin{align*}
D^{(1)}(1) = B(1), \ D^{(1)}(2) = B(2), \ D^{(1)}(3) = B(3) +B(4), \\
D^{(2)}(1) = B(3), \ D^{(2)}(2) = B(4), \ D^{(2)}(3) = B(1) +B(2),
\end{align*}
since then the effects $\id, D^{(1)}(1)$, and $D^{(2)}(1)$ are
linearly independent and $B(i) \in  \mathrm{conv}\left(
\{\{D^{(j)}(k) \}_{j,k}, O, \id \} \right)$ for all $i=1,2,3,4$
but now by the same arguments as above we have that $\B \notin
\simu{\{\D^{(1)},\D^{(2)}\}}$. This also shows that
Proposition~\ref{prop:conv-sim} does not hold with simulated observables
which have more than two outcomes.
\end{example}

If we restrict ourselves to sets of simulators composed of
dichotomic observables, the converse of Proposition~\ref{prop:conv-sim}
is seen to hold even when allowing more outcomes for the simulated
observables.

\begin{proposition}\label{prop:sim-conv}
Let $\mathcal{B}=\{\B^{(i)}\}_{i=1}^m$ be a collection of $m$ dichotomic observables. For an observable $\A$ with an outcome set $\Lambda$, the following implication holds:
$$
\A \in \simu{\mathcal{B}} \quad \Rightarrow \quad \A_y \in
\mathrm{conv}\left(\{\{\B^{(i)}_\pm\}_{i},o,u\}\right) \quad \forall y
\in \Lambda.
$$
\end{proposition}
\begin{proof}
Denote $\mathcal{I}_m = \{1, \ldots,m\}$. Let $\A \in \simu{\mathcal{B}}$ so that
\begin{equation}
\A_y = \sum_{i \in \mathcal{I}_m} p_i\left( \nu_{(i,+)y} \B^{(i)}_+ + \nu_{(i,-)y} \B^{(i)}_- \right)
\end{equation}
for some probability distribution $(p_i)_{i=1}^m$ and a postprocessing $\nu: \mathcal{I}_m \times \{+,-\} \to \Lambda$.

For each $y \in \Lambda$ we denote $\mathcal{I}_y^+ = \{ i \in
\mathcal{I}_m\, | \, \nu_{(i,+)y} \geq \nu_{(i,-)y} \}$ and
$\mathcal{I}_y^- = \mathcal{I}_m \setminus \mathcal{I}_y^+$. Now
we may express each effect $\A_y$ as
\begin{eqnarray*}
\A_y &=& \sum_{i \in \mathcal{I}_y^+} p_i \left( \nu_{(i,+)y}\B^{(i)}_+ + \nu_{(i,-)y} \B^{(i)}_- \right)\\
  & &+ \sum_{i \in \mathcal{I}_y^-} p_i \left( \nu_{(i,+)y}\B^{(i)}_+ + \nu_{(i,-)y} \B^{(i)}_- \right) \\
  &=& \sum_{i \in \mathcal{I}_y^+} p_i \left[( \nu_{(i,+)y} -\nu_{(i,-)y} ) \B^{(i)}_+ + \nu_{(i,-)y} u \right] \\
  & & + \sum_{i \in \mathcal{I}_y^-} p_i \left[( \nu_{(i,-)y} -\nu_{(i,+)y} ) \B^{(i)}_- + \nu_{(i,+)y} u \right] \\
  &=& \sum_{i \in \mathcal{I}_y^+} p_i ( \nu_{(i,+)y} -\nu_{(i,-)y} ) \B^{(i)}_+  \\ & & + \sum_{i \in \mathcal{I}_y^-} p_i ( \nu_{(i,-)y} -\nu_{(i,+)y} ) \B^{(i)}_- \\
  & & + \left[\sum_{i \in \mathcal{I}_y^+} p_i \nu_{(i,-)y}
 + \sum_{i \in \mathcal{I}_y^-} p_i \nu_{(i,+)y} \right] u,
\end{eqnarray*}
where we have used the fact that $\B^{(i)}_-=u-\B^{(i)}_+$ for all $i\in \mathcal{I}_m$. We see that now the coefficients of all the effects in the above expression are positive and for the total sum of the coefficients we have that
\begin{eqnarray*}
& &\sum_{i \in \mathcal{I}_y^+} p_i ( \nu_{(i,+)y} -\nu_{(i,-)y} )  + \sum_{i \in \mathcal{I}_y^-} p_i ( \nu_{(i,-)y} -\nu_{(i,+)y} ) \\
& & + \sum_{i \in \mathcal{I}_y^+} p_i \nu_{(i,-)y}
 + \sum_{i \in \mathcal{I}_y^-} p_i \nu_{(i,+)y} \\
 &=& \sum_{i \in \mathcal{I}_y^+} p_i \nu_{(i,+)y}
 + \sum_{i \in \mathcal{I}_y^-} p_i \nu_{(i,-)y} \\
 &\leq & \sum_{i \in \mathcal{I}_y^+} p_i
 + \sum_{i \in \mathcal{I}_y^-} p_i  = \sum_{i \in \mathcal{I}_m} p_i =1.
\end{eqnarray*}
Thus, by adding the zero effect $o$ in the last expression for $\A_y$ with a weight of $1- \sum_{i \in \mathcal{I}_y^+} p_i \nu_{(i,+)y}
 - \sum_{i \in \mathcal{I}_y^-} p_i \nu_{(i,-)y}$ we get a convex decomposition for $\A_y$ so that
\begin{equation}\label{eq:conv-hull}
\A_y \in \mathrm{conv}\left(\{\{\B^{(i)}_\pm \}_{i},o,u\}\right)
\end{equation}
for all $y \in \Lambda$.
\end{proof}

The previous proposition shows that if an observable is effectively dichotomic, all of its effects are contained in the convex hull of the zero effect, the unit effect, and the effects of the dichotomic simulator observables. That is, if for a given set of dichotomic observables corresponding to some measurement devices in a laboratory, we choose some postprocessing and a probability distribution such that we make a simulation with those measurement devices, the previous proposition can be used to extract the simulated observable's convex decomposition into the effects of the set of simulators and the zero and the unit effect, thereby giving us their mathematical expressions.

On the other hand, it gives a useful necessary condition for
dichotomic simulability in an experimental setting. Let us say we
have access to some fixed set of measurement devices that
correspond to some dichotomic observables $\mathcal{B}$ and we
want to know whether a given observable $\A$ can be simulated
using the accessible measurements. If we find an effect of $\A$
that is not contained in the convex hull of $o$, $u$, and the
effects of the observables in $\mathcal{B}$, we know that $\A$
cannot be simulated by $\mathcal{B}$.

In general, however, we note that if the set of simulators is not fixed, for any observable we can always find such dichotomic observables so that condition \eqref{eq:conv-hull} is satisfied, namely the binarizations of the given observable.

From Propositions~\ref{prop:conv-sim-lin-ind} and \ref{prop:sim-conv}, we get the following corollary.

\begin{corollary}\label{cor:lin-ind}
Let $\mathcal{B}=\{\B^{(i)}\}_{i=1}^m$ be a collection of $m$ dichotomic observables such that the set $\left\{ u, \{\B^{(i)}_+\}_{i=1}^m \right\}$ is linearly independent. An observable $\A$ with an outcome set $\Lambda$ is contained in $\simu{\mathcal{B}}$ if and only if $\A_y \in \mathrm{conv}\left(\{\{\B^{(i)}_\pm\}_{i},o,u\}\right)$ for all outcomes $y \in \Lambda$.
\end{corollary}

If the set of simulators as well as the simulated observable are
all dichotomic we get the following simple corollary from
Propositions~\ref{prop:conv-sim} and \ref{prop:sim-conv}:

\begin{corollary}\label{cor:dichotomic}
Let $\mathcal{B}=\{\B^{(i)}\}_{i=1}^m$ be a collection of $m$ dichotomic observables. A dichotomic observable $\A$ is cointained in $\simu{\mathcal{B}}$ if and only if $\A_+ \in \mathrm{conv}\left(\{\{\B^{(i)}_\pm\}_{i},o,u\}\right)$.
\end{corollary}

From Propositions~\ref{prop:conv-sim-single-lin-ind} and
\ref{prop:sim-conv}, we get a full characterization for the
simulation set of a single simulation irreducible dichotomic
observable.

\begin{corollary}
Let $\B$ be a simulation irreducible dichotomic observable. An observable $\A$ with an outcome set $\Lambda$ is contained in $\simu{\B}$ if and only if $\A_y \in \mathrm{conv}\left(\{\B_+,\B_-,o,u\}\right)$ for all $y \in \Lambda$.
\end{corollary}

\begin{example} \label{ex:qubit}
A qubit effect $E$ can be written in the form $E=\half
\left[(1+e_0)\id + \ve \cdot \vsigma\right]$ for some
$e_0\in\real$ and $\ve = (e_x,e_y,e_z) \in \real^3$ satisfying
$|e_0|+\no{\ve}_2 \leq 1$. The real number $e_0 \in [-1,1]$ is
called the \emph{bias} of the effect $E$, with $E$ being
\emph{unbiased} if $e_0=0$. We denote by $\X$ ,$\Y$ and $\Z$ the
observables that have the effects $X(\pm)= \half (\id \pm
\sigma_x)$, $Y(\pm) = \half (\id \pm \sigma_y)$ and $Z(\pm) =
\half (\id \pm \sigma_z)$, and consider the  simulation set
$\simu{\X,\Y,\Z}$ of those observables. We also denote by $\T$ the
trivial observable with effects $T(+) =\id$ and $T(-) =O$.

Since the set of effects $\{\id, X(+),Y(+), Z(+)\}$ is linearly
independent, it follows from Corollary \ref{cor:lin-ind} that a qubit
observable $\E$ with an outcome set $\Omega$ is contained in
$\simu{\X,\Y,\Z}$ if and only if the effects $E(j)= \half
\left[(1+e^{(j)}_0)\id + \vec{e}^{(j)} \cdot \vsigma\right] \in
\mathrm{conv} \left( \{X(\pm),Y(\pm),Z(\pm),O,\id\} \right)$ for
all $j \in \Omega$. The set of effects
$\{T(\pm),X(\pm),Y(\pm),Z(\pm)\}$ is convexly independent, so 
the set of extreme effects of
$\mathrm{conv}\left(\{T(\pm),X(\pm),Y(\pm),Z(\pm)\}\right)$ are exactly the
effects $\{T(\pm),X(\pm),Y(\pm),Z(\pm)\}$. These effects
correspond to vectors $\{(\pm 1,0,0,0),(0,\pm 1,0,0),(0,0,\pm
1,0),(0,0,0,\pm 1)\}$ in $\real^4$, respectively, which in turn
are the extreme points of the four-dimensional convex set
\begin{equation}
S^4 = \left\lbrace (r_0,r_1,r_2,r_3) \in \real^4 \, \mid \, \sum_{i=0}^3 |r_i| \leq 1 \right\rbrace.
\end{equation}
Thus, there is a one-to-one correspondence with the effects in
$\mathrm{conv}\left(\{X(\pm),Y(\pm),Z(\pm),O,\id\}\right)$ and the points in
$S^4$, and so the observable $\E$ with effects $E(j) =  \half
\left[(1+e^{(j)}_0)\id + \vec{e}^{(j)} \cdot \vsigma\right]$ is in
$ \simu{\X,\Y,\Z}$ if and only if $(e^{(j)}_0, \vec{e}^{(j)}) \in
S^4$ for all $j \in \Omega$, i.e.,
\begin{equation}\label{eq:qubit-L1}
|e^{(j)}_0| + \no{\ve^{(j)}}_1 \leq 1.
\end{equation}

\begin{figure}[t]
\includegraphics[width=0.3\textwidth]{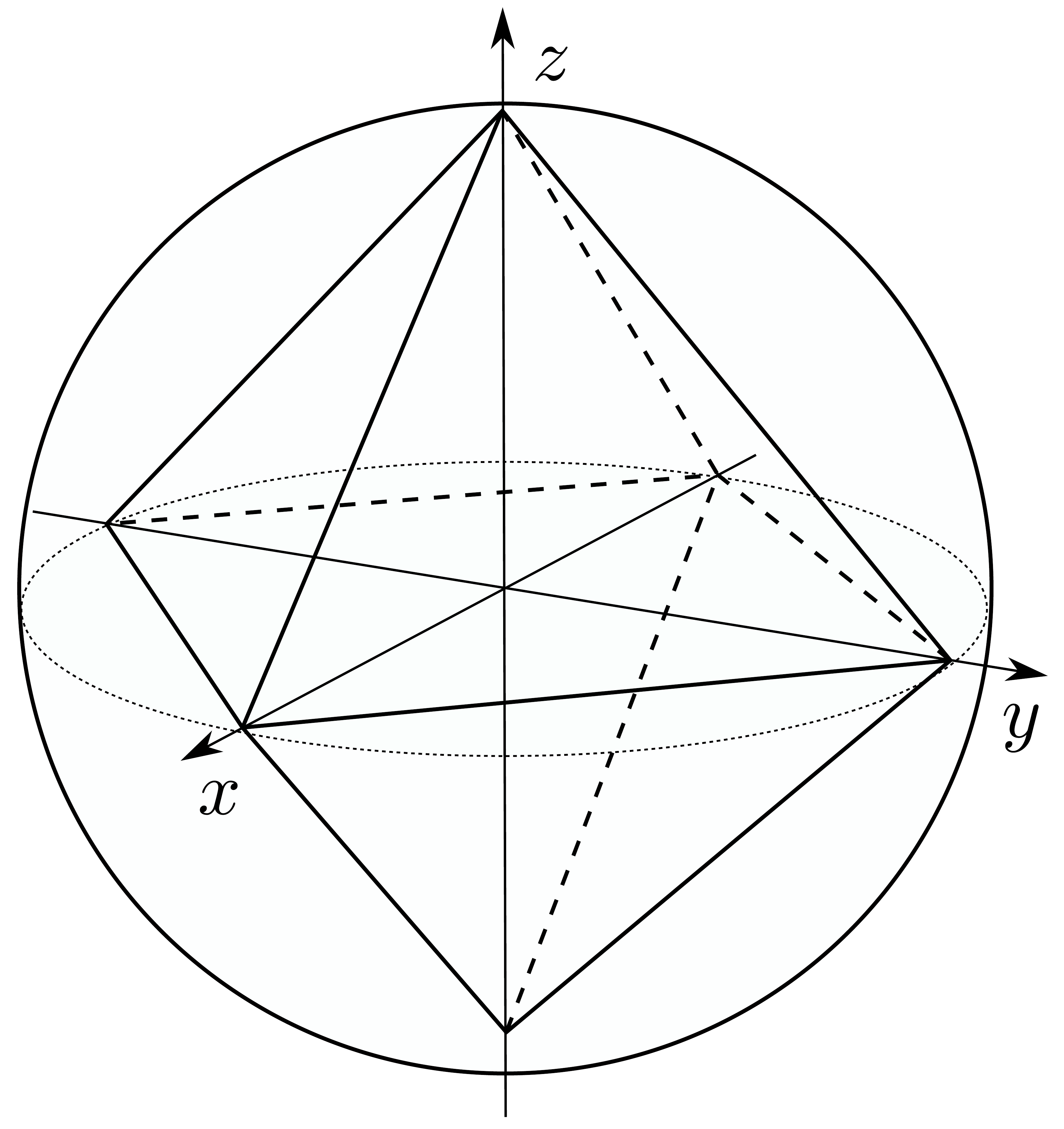}
\caption{\label{figure6} The unbiased effects in $\simu{\X,\Y,\Z}$
form an octahedron.}
\end{figure}

For the unbiased case, i.e., when $e^{(j)}_0 =0$ for all $j \in
\Omega$, inequality \eqref{eq:qubit-L1} defines an octahedron in
$\real^3$ which is depicted in Fig.~\ref{figure6}. We also see
that the set of unbiased effects in $\simu{\X,\Y}$ forms a square
in $\real^2$ (as do $\simu{\X,\Z}$ and $\simu{\Y,\Z}$ too).
\end{example}

\section{Nonquantum state spaces}\label{sec:non-quantum}

\subsection{Classical state spaces}

A state space $\state$ is \emph{classical} if all pure states are
distinguishable, or equivalently, $\state$ is simplex. Up to the
labeling of outcomes, the observable that can distinguish all
pure states is unique. It is clear that any classical state space
$\state_{cl}$ has only one equivalence class of simulation irreducible observables: Let $\G$ be the observable on $\state_{cl}$ that distinguishes the pure states of $\state_{cl}$. For each observable $\A$ we define a
postprocessing $\nu^\A$ by setting $\nu^\A_{xy} = \A_y(s_x)$ for
all outcomes $y$ and pure states $s_x \in \state_{cl}^{ext}$.
Since for any state $s = \sum_x \lambda_x s_x$ we have that
$\G_x(s) = \lambda_x$, and so
\begin{equation}
\A_y(s) = \sum_x \lambda_x \A_y(s_x) = \sum_x \nu^\A_{xy} \G_x(s) = (\nu^\A \circ \G)_y(s)
\end{equation}
for all outcomes $y$ and states $s \in \state_{cl}$, so $\A
\in \simu{\G}$ and therefore $\obs = \simu{\G}$. If $\G'$ is some other simulation irreducible observable, then $\G' \in \simu{\G}$, and from the fact that $\G'$ is postprocessing clean it follows that also $\G \in \simu{\G'}$ which yields $\G' \leftrightarrow \G$. Furthermore, the extreme simulation irreducible
observable has the same number of outcomes as the number of pure
states in $\state_{cl}$. We conclude that the effective number of
any observable in a classical state space is at most $n$, where
$n$ is the number of pure states.

On the other hand, if there exists only a single equivalence class of simulation
irreducible observables on a state space $\state$, the state space
must be classical; this follows from the result of Ref. 
\cite{Plavala16}. In the following, we give an alternative proof of
this fact, relying on the properties of simulation irreducible
observables. 

Let us denote $d =\dim({\rm aff}(\state))$ so that as
in Sec. \ref{sec:preliminaries} we can consider $\state$ and
$\effect(\state)$ to be embedded in $(d+1)$-dimensional ordered
vector spaces $\mathcal{A}$ and $\mathcal{A}^*$ respectively.
Denote by $\B$ the extreme simulation irreducible observable in the equivalence class and
suppose it has $n$ outcomes. From Proposition~\ref{prop:finite}, it
follows that every observable on $\state$ can be simulated with
$\B$. Now $\B$ consists of $n$ linearly
independent indecomposable effects $\B_i$. For each indecomposable
effect $\B_i$ there exists an extreme effect $b_i$ and $\beta_i
\in (0,1]$ such that $\B_i = \beta_i b_i$ for all $i =1,
\ldots,n$ \cite{KiNuIm10}. Since the $n$ dichotomic observables determined by the
effects $b_i$ must be simulable by $\B$, there exist
postprocessings $\nu^{(i)}$ such that $b_i = \sum_j \nu_{j+} \B_j
= \sum_j \nu_{j+} \beta_j b_j$ for all $i = 1, \ldots,n$ and since
the set $\{b_i\}_i$ is linearly independent, it follows that
$\beta_i =1$ for all $i = 1, \ldots,n$. Thus, the effects of $\B$
are actually extreme.

It is easy to see that for each extreme effect there exists an extreme state that gives probability one for the state \cite{KiNuIm10}. Thus, for every effect $\B_i$ there exists a pure state $s_i$ such that $\B_i(s_i)=1$ for all $i=1, \ldots,n$. Furthermore, due to the normalization of $\B$, we have that
\begin{equation*}
1 = u(s_i) = \sum_j \B_j(s_i) = \B_i(s_i) + \sum_{j \neq i} \B_j(s_i) = 1+ \sum_{j \neq i} \B_j(s_i)
\end{equation*}
so that $\B_j(s_i)=0$ and $s_j \neq s_i$ for all $j\neq i$ where $i=1,\ldots,n$. Hence, $\B$ distinguishes the set of states $\{s_1, \ldots,s_n\}$.

We now note that the effects of $\B$ are the only indecomposable
effects that lie on different extreme rays. Indeed, let $e$ be any
indecomposable effect and consider the dichotomic observable $\E$
with $\E_+ =e$. Since $\E \in \simu{\B}$, there exists a
postprocessing $\mu$ such that $e = \sum_i \mu_{i+} \B_i$ so that
from the indecomposability of $e$ it follows that $e$ is
proportional to $\B_l$ for some $l \in \{1, \ldots,n\}$. Thus,
there exist exactly $n$ linearly independent extreme rays that define the 
generating positive cone in the $(d+1)$-dimensional effect space, and therefore we
must have that $n=d+1$.

It is straightforward to check that the states $\{s_1, \ldots,s_n\}$ are affinely independent so that $\dim( {\rm aff }( \{s_1,\ldots,s_n\})) = n-1 =d = \dim({\rm aff}(\state))$. Thus, every state $s \in \state$ can be expressed as an affine combination of the states $\{s_1, \ldots,s_n\}$, i.e., $s = \sum_i \gamma_i s_i$ for some $\{\gamma_i\}_i \subseteq \real$ such that $\sum_i \gamma_i =1$. However, we see that
\begin{equation*}
\gamma_j= \sum_i \gamma_i \B_j(s_i) = \B_j(s) \geq 0
\end{equation*}
so the affine decomposition of $s$ is actually convex,
which shows that the only pure states are actually $s_1,
\ldots,s_n$. Since $\state$ is then a convex hull
of $d+1$ affinely independent (distinguishable) pure states,
$\state$ must be a $d$-simplex.

We can rephrase this result as follows.

\begin{proposition}
A state space is nonclassical if and only if there exist at least two inequivalent simulation irreducible observables.
\end{proposition}

\subsection{Square bit state space}

Consider a state space $\state_\square = {\rm
conv}(\{s_1,s_2,s_3,s_4\})$ that is isomorphic to a square in
$\mathbb{R}^2$, i.e., $s_1 + s_3 = s_2 + s_4$, (see Fig.~\ref{figure7}).
Such a state space is also referred to as the square bit state
space or squit state space. The set of effects
$\effect(\state_\square)$ is an intersection of the positive dual
cone $\mathcal{A}_+^*$ and the set $u-\mathcal{A}_+^*$,
which is isomorphic to the octahedron in $\mathbb{R}^3$,
Fig.~\ref{figure7}.

In this section we demonstrate that the set of all observables
$\obs$ on the square bit state space can be simulated from a set
of two binary observables $\E$ and $\F$ defined as follows:
\begin{eqnarray*}
&& \E_+(s_1) = \E_+(s_2) = 0, \quad \E_+(s_3) = \E_+(s_4) = 1, \\
&& \E_-(s_1) = \E_-(s_2) = 1, \quad \E_-(s_3) = \E_-(s_4) = 0.
\end{eqnarray*}
\begin{eqnarray*}
&& \F_+(s_1) = \F_+(s_4) = 0, \quad \F_+(s_2) = \F_+(s_3) = 1, \\
&& \F_-(s_1) = \F_-(s_4) = 1, \quad \F_-(s_2) = \F_-(s_3) = 0.
\end{eqnarray*}

\begin{figure}[b]
\includegraphics[scale=0.27]{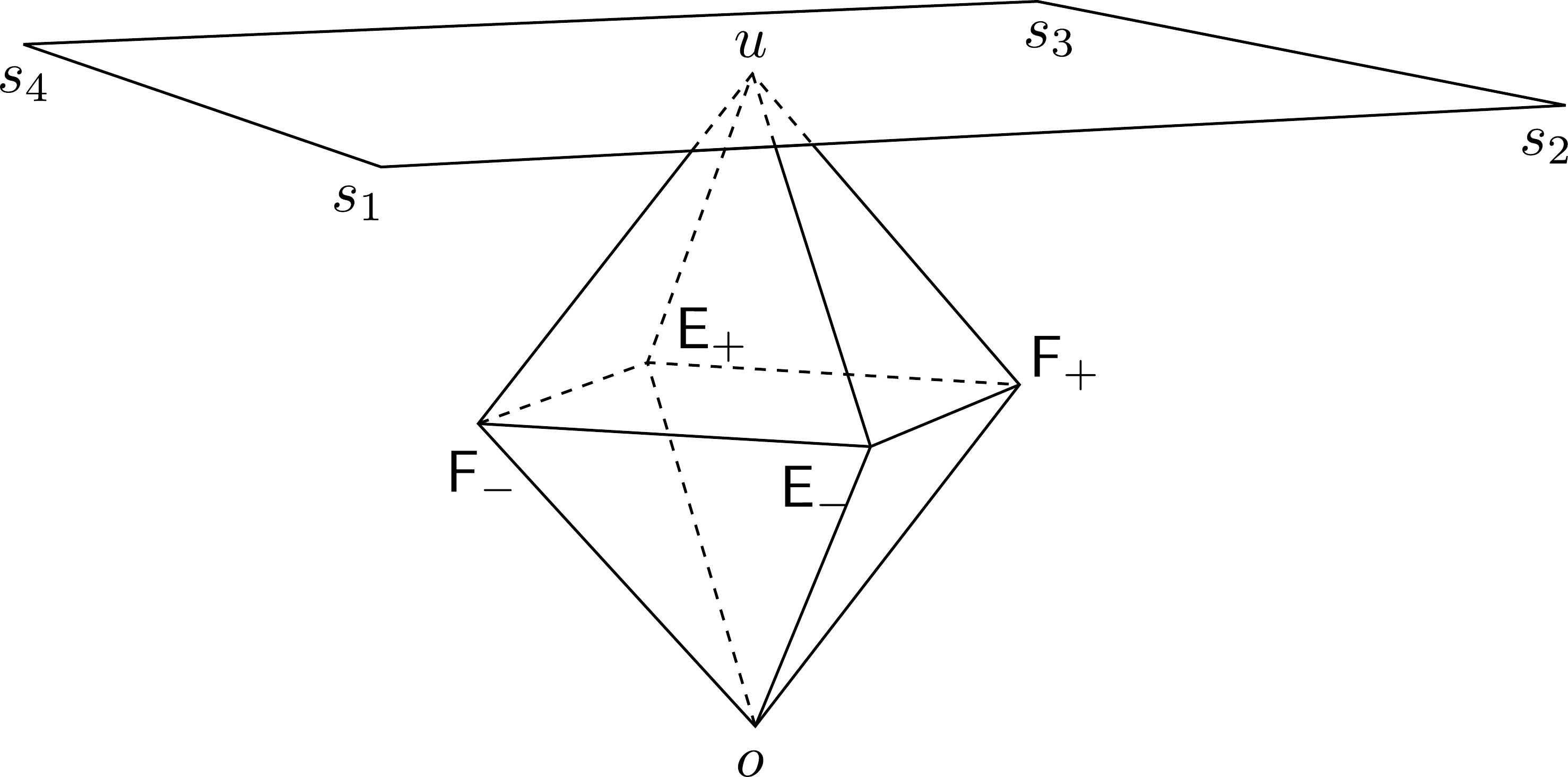} \caption{\label{figure7} Square state space (above) and the space of effects (below).}
\end{figure}

Since the set of effects $\{u,\E_+,\F_+\}$ is linearly
independent, it follows from Corollary \ref{cor:lin-ind} that an
observable $\A$ with outcome set $\Omega$ is contained in
$\simu{\{\E,\F\}}$ if and only if $\A_x \in {\rm
conv}\left(\{\E_+,\E_-,\F_+,\F_-,o,u\}\right)$ for all $x \in \Omega$, which is
always fulfilled because ${\rm conv}\left(\{\E_+,\E_-,\F_+,\F_-,o,u\}\right) =
\effect(\state_\square)$. Hence, $\simu{\{\E,\F\}}= \obs$.

The obtained result implies the following:
\begin{itemize}
\item The effective number of outcomes for any observable on the
square bit state space is at most 2.

\item Any simulation irreducible observable is postprocessing
equivalent to either $\E$ or $\F$.

\item $\smin{\obs}=2$.
\end{itemize}

It is known that the square bit state space possesses the feature of maximal incompatibility: There exists a pair of observables (which are actually exactly the observables $\E$ and $\F$) such that the minimum amount of noise one has to mix them with to make their noisy versions compatible is enough to make any other pair of observables compatible in any theory \cite{BuHeScSt13}. In this sense, the square bit state space is even more nonclassical than any finite-dimensional quantum theory \cite{HeScToZi14}.

Since classical theories have only one equivalence class of simulation irreducible observables, we can argue that theories, such as square bit state space, having just two of such equivalence classes are somewhat closest to classical theory. Furthermore, the effective number of all observables on this state space is the same as in the simplest and one of the most important classical theories, namely the bit. In this sense, the square state space is closest to classical theory amongst all nonclassical theories.

\subsection{Polygon state spaces}

We say that a convex set $P_n$ is a regular $n$-sided polygon if
there exist $n$ vectors $\vec{p}_1, \ldots, \vec{p}_n$ in $\real^2$ such that $\no{\vec{p}_1} =
\no{\vec{p}_2} = \ldots = \no{\vec{p}_n}$, and $\vec{p}_i \cdot
\vec{p}_{i+1} = \no{\vec{p}_i}^2 \cos{\left(\frac{2
\pi}{n}\right)}$ for all $i=1, \ldots,n$ (where the addition is
modulo $n$) such that $P_n$ is isomorphic to
$\mathrm{conv}\left(\{\vec{p}_1, \ldots,\vec{p}_n\}\right)$. The
extremal points of a polygon are its vertices, and faces are
exactly the sides of the polygon; see Fig.~\ref{figure8}.

As a state space, we consider polygons embedded in $\real^3$ lying on the $z=1$ plane. A polygon state space $\mathcal{S}_n$ with $n$ vertices is then given by the convex hull of $n$ extremal states
\begin{equation}\label{eq:poly-state}
\vec{s}_k=
\begin{pmatrix}
\sec\left(\dfrac{\pi}{n}\right) \cos\left(\dfrac{2 k \pi }{n}\right) \\
\sec\left(\dfrac{\pi}{n}\right) \sin\left(\dfrac{2 k \pi }{n}\right) \\
1
\end{pmatrix}, \quad k = 1,\ldots,n.
\end{equation}

As the polygons are two-dimensional, the effects can also be
represented as elements in $\real^3$. Hence, we can express each
$e \in \mathcal{E}(\mathcal{S}_n)$ as $\vec{e} =\left( a_x,a_y,a_z
\right)^T \in \real^3$. With this identification we have that
$e(s)=\vec{e} \cdot \vec{s}$ for all $e \in
\mathcal{E}(\mathcal{S}_n)$ and $s \in \mathcal{S}_n$, where now
$\vec{e},\vec{s} \in \real^3$ and $\cdot$ is the Euclidean dot
product in $\real^3$. We omit the vector notation from here
onwards and simply denote the states and effects in $\real^3$ by
$s$ and $e$ instead of $\vec{s}$ and $\vec{e}$. Clearly, we now have the zero
effect $o=(0,0,0)^T$ and the unit effect $u=(0,0,1)^T$.

To find the positive dual cone $\mathcal{A}_+^*=\{e \, | \, e(s)
\geq 0$ for all $s \in \mathcal{S}_n \}$, it is enough to satisfy
the requirement $e(s_k) \geq 0$ for all extremal states
\eqref{eq:poly-state}. We have
\begin{equation}
e(s_k) = a_x \sec\left(\dfrac{\pi}{n}\right) \cos{\left(\dfrac{2 k \pi }{n}\right)} + a_y \sec\left(\dfrac{\pi}{n}\right)
\sin{\left(\dfrac{2 k \pi }{n}\right)} + a_z \geq 0,
\end{equation}

\noindent $k = 1, \ldots, n$. The extremal rays of the positive
dual cone $\mathcal{A}_+^*$ correspond to the intersection of two
adjacent planes $e(s_k)=0$ and $e(s_{k-1})=0$ and have the form
\begin{equation} \label{eq:rays}
e_k^+ = \begin{pmatrix}
-a_z \cos\left(\dfrac{(2 k-1) \pi}{n}\right) \\
-a_z \sin\left(\dfrac{(2 k-1) \pi}{n}\right) \\
a_z
\end{pmatrix}, \quad k =1,\ldots,n, \quad a_z \geq 0.
\end{equation}

\noindent Similarly, inequalities $e(s_k) \leq 1$, $k=1,\ldots,n$
define the set $u-\mathcal{A}_+^*$ with extremal rays
\begin{equation}
e_k^- = \begin{pmatrix}
b_z \cos\left(\dfrac{(2 k-1) \pi}{n}\right) \\
b_z \sin\left(\dfrac{(2 k-1) \pi}{n}\right) \\
1 - b_z
\end{pmatrix}, \quad k =1,\ldots,n, \quad b_z \geq 0.
\end{equation}

If $n$ is even, then the extremal rays $e_{k+n/2}^+$ and $e_k^-$
intersect, with the resulting nontrivial extremal effects being
\begin{equation}\label{eq:poly-effect}
e_k = \dfrac{1}{2}
\begin{pmatrix}
\cos{\left(\dfrac{(2 k-1) \pi}{n}\right)} \\
 \sin{\left(\dfrac{(2 k-1) \pi}{n}\right)} \\
1
\end{pmatrix}, \quad k = 1,\ldots,n.
\end{equation}

\begin{figure}
\includegraphics[scale=0.25]{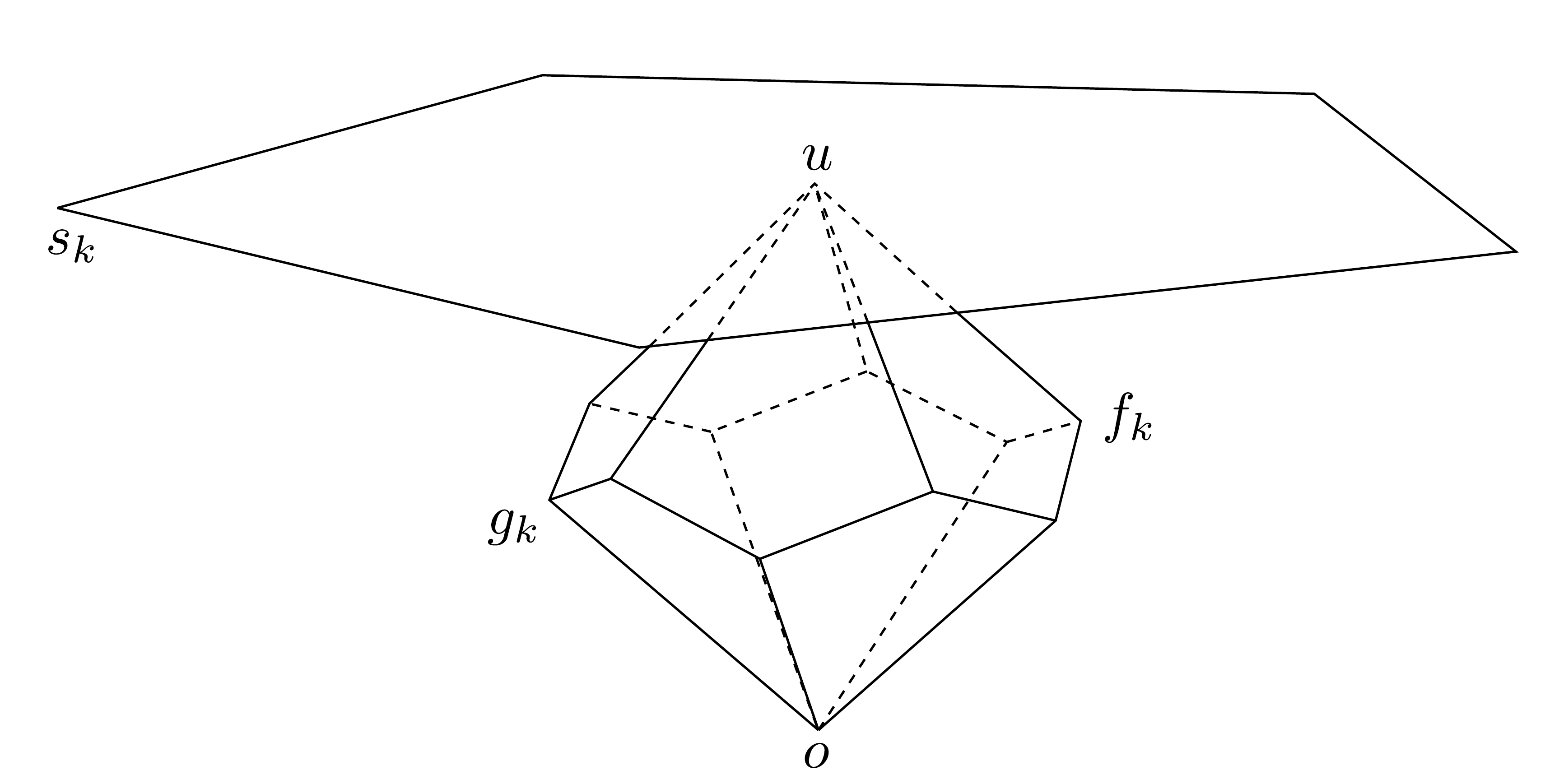} \\ \hspace*{0.4cm}
\includegraphics[scale=0.25]{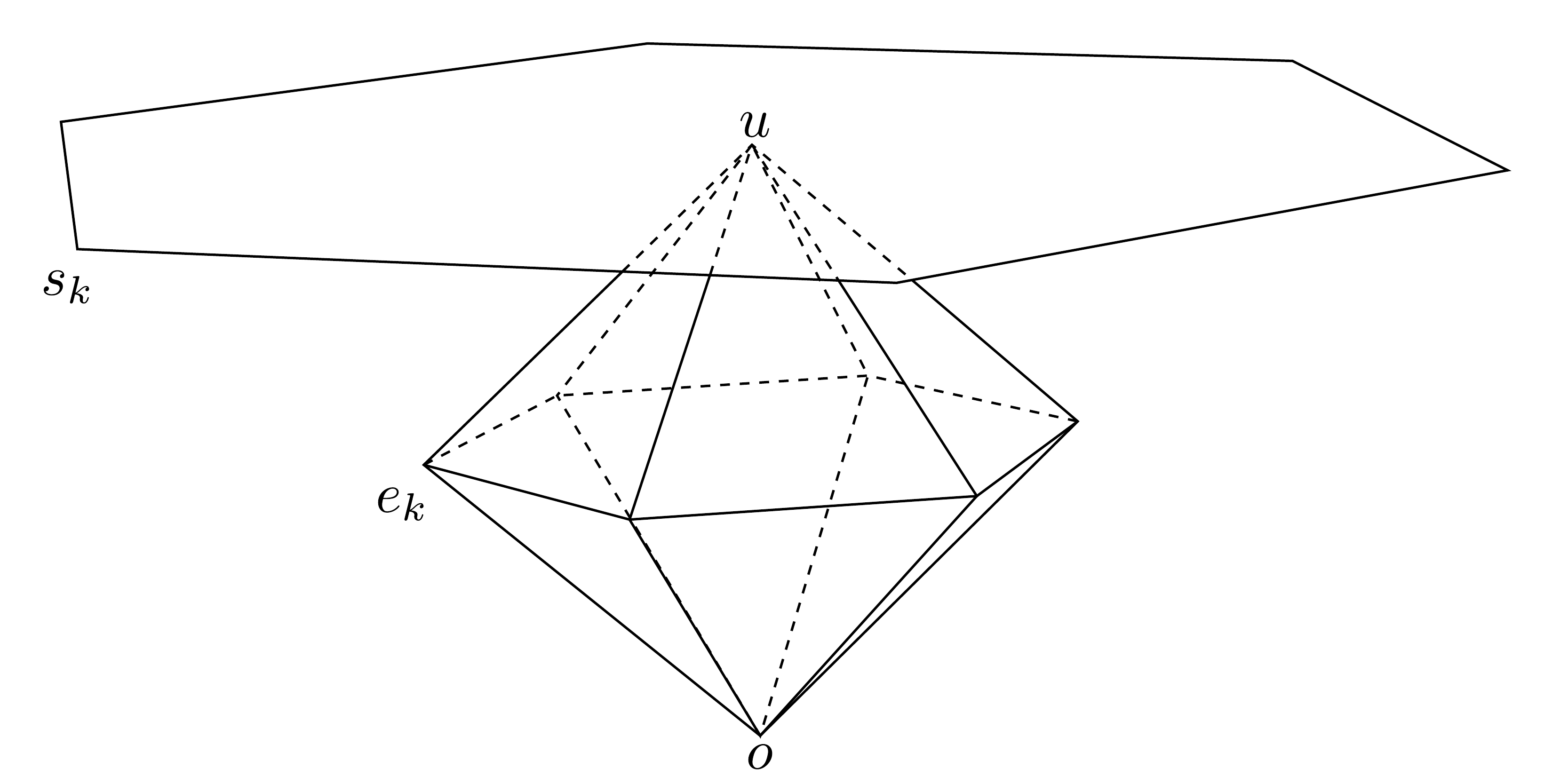}
\caption{\label{figure8} Odd and even polygon state spaces
$\mathcal{S}$ and corresponding sets of effects
$\mathcal{E}(\mathcal{S})$.}
\end{figure}

If $n$ is odd, then the rays $e_k^+$ and $e_{k'}^-$ do not
intersect. In this case, the intersection of upward and downward
cones results in two families of extremal effects. The first
family corresponds to points at which $\frac{1}{2}e^+_{k+(n-1)/2}
+ \frac{1}{2}e^+_{k+(n-1)/2-1} = e_k^-$ and reads
\begin{equation}\label{eq:poly-effect-odd-1}
f_k=\dfrac{1}{1+\sec\left(\dfrac{\pi}{n}\right)}
\begin{pmatrix}
\cos{\left(\dfrac{(2 k-1) \pi}{n}\right)} \\
 \sin{\left(\dfrac{(2 k-1) \pi}{n}\right)} \\
\sec\left(\dfrac{\pi}{n}\right)
\end{pmatrix}, \quad k =1,\ldots,n.
\end{equation}

\noindent The second family corresponds to points at which
$\frac{1}{2}e^-_{k+(n-1)/2} + \frac{1}{2}e^-_{k+(n-1)/2-1} = e_k^+$
and reads
\begin{equation}\label{eq:poly-effect-odd-2}
g_k=\dfrac{1}{1+\sec\left(\dfrac{\pi}{n}\right)}
\begin{pmatrix}
-\cos{\left(\dfrac{(2 k-1) \pi}{n}\right)} \\
-\sin{\left(\dfrac{(2 k-1) \pi}{n}\right)} \\
1
\end{pmatrix} = u - f_{k}, \quad k =1,\ldots,n.
\end{equation}

\noindent In this case of odd $n$ we note that the nontrivial
extremal effects no longer lie in a single plane; see
Fig.~\ref{figure8}.

Thus, in the case of even polygon state spaces we have
$\mathcal{E}(\mathcal{S}_n) = \mathrm{conv}\left(\{e_1,
\ldots,e_n,o,u\}\right) = \mathrm{conv}\left(\{\E_{\pm}^{(1)},
\ldots,\E_{\pm}^{(n/2)},o,u\}\right)$, where we have defined the
dichotomic observables $\E^{(i)}$ with effects $\E_+^{(i)} = e_i$
and $\E_{-}^{(i)} = u - \E_+^{(i)} = e_{i+n/2}$, $i = 1, \ldots,
\frac{n}{2}$. In the case of odd polygon state spaces, we have
$\mathcal{E}(\mathcal{S}_n) =
\mathrm{conv}\left(\{f_1,\ldots,f_n,g_1,\ldots,g_n,o,u\}\right) =
\mathrm{conv}\left(\{\F_{\pm}^{(1)},\ldots,\F_{\pm}^{(n)},o,u\}\right)$,
where we have defined the dichotomic observables $\F^{(i)}$ with
effects $\F_+^{(i)} = f_i$ and $\F_{-}^{(i)} = u - \F_+^{(i)} =
g_i$, $i = 1, \ldots, n$.

The fundamental difference between the effect spaces for even and
odd polygon state spaces is that, in the case of even $n$, to
construct $\mathcal{E}(\mathcal{S}_n)$ one needs the effects of
$\frac{n}{2}$ dichotomic observables (plus the zero and the unit
effect), whereas in the case of odd $n$, one needs the effects of $n$
dichotomic observables (plus the zero and the unit effect) to get
the whole effect space $\mathcal{E}(\mathcal{S}_n)$.

However, we find that Proposition~\ref{prop:conv-sim} has strong
consequences in polygon state spaces in both even and odd cases.
Namely, if $\A$ is a dichotomic observable on a polygon state
space $\mathcal{S}_n$ with $n$ vertices, then always
\begin{equation}
\A_+ \in
\begin{cases}
\mathrm{conv}\left(\{\E^{(1)}_\pm, \ldots,\E^{(\frac{n}{2})}_\pm ,o,u\}\right) & {\rm if \ } n {\rm \ is \ even}, \\
\mathrm{conv}\left(\{\F^{(1)}_\pm, \ldots,\F^{(n)}_\pm
,o,u\}\right) & {\rm if \ } n {\rm \ is \ odd}.
\end{cases}
\end{equation}
From Proposition~\ref{prop:conv-sim} it follows that
\begin{equation}
\A \in
\begin{cases}
\simu{\{\E^{(1)}, \ldots,\E^{(\frac{n}{2})}\}} & {\rm if \ } n {\rm \ is \ even}, \\
\simu{\{\F^{(1)}, \ldots,\F^{(n)}\}} & {\rm if \ } n {\rm \ is \
odd},
\end{cases}
\end{equation}
so that for the set $\mathcal{O}_\pm$ of all dichotomic
observables on $\mathcal{S}_n$ we have
\begin{equation}
\smin{\mathcal{O}_\pm} \leq \begin{cases}
\frac{n}{2} & {\rm if \ } n {\rm \ is \ even}, \\
n & {\rm if \ } n {\rm \ is \ odd}.
\end{cases}
\end{equation}

Next, we will characterize the extreme simulation irreducible observables in polygon state spaces.

\begin{proposition} \label{prop:poly-even}
The minimal simulation number for the set $\obs$ of all
observables on an even polygon state space $\mathcal{S}_{2m}$ equals
$\smin{\obs} = m+\frac{1}{3}m(m-1)(m-2)$.
\end{proposition}
\begin{proof}
From Proposition~\ref{prop:finite} it follows that in order to find
$\smin{\obs}$ one merely needs to know the number of
inequivalent simulation irreducible observables. By
Corollary~\ref{cor:irr-indec-extreme}, it is enough to find the number
of inequivalent observables $\A$ with linearly independent
indecomposable effects. Since $\A$ is indecomposable, its effects
belong the extreme rays of the positive effects cone, i.e., they
are some positive scalar multiples of the nontrivial extremal
effects $e_k$ in \eqref{eq:poly-effect}. Furthermore, since the
effects of $\A$ are linearly independent and contained in
$\real^3$, $\A$ has at most three outcomes.

If $\A$ is dichotomic, then the only possibility is that $\A_+ =
e_k$ and $\A_- = e_{k+m}$, $k=1,\ldots,2m$. Thus, there are $2m$
choices for the effects of $\A$. Taking into account the bijective
relabellings of outcomes, i.e., the permutations of the set
$\{+,-\}$, we have $2m/2! = m$ inequivalent simulation
irreducible dichotomic observables.

If $\A$ is trichotomic with effects $\A_1$, $\A_2$, and $\A_3$,
then $\A_j= c_j e_{k_j}$ for some $k_j \in \{1,\ldots,2m\}$ and
$0<c_j\leq 1$ for all $j=1,2,3$ such that $k_1\neq k_2 \neq
k_3\neq k_1$. Denote $c\equiv\sum_{j=1}^3 c_j \neq 0$, and then from the
normalization of $\A$ it follows that
\begin{equation}\label{eq:poly-conv-hull}
\sum_{j=1}^3 \dfrac{c_j}{c} e_{k_j} = \dfrac{1}{c} u.
\end{equation}
Since $\frac{1}{2}u$ is the only scalar multiple of $u$ contained
in the plane of nontrivial extreme effects, with necessity $c=2$.
Therefore, $\frac{1}{2}u$ must be contained in the convex hull of
the extreme effects $\{e_{k_j}\}_{j=1}^3$ which limits the choices
of the indices $k_j$. Moreover, since the convex hull of the three
effects $e_{k_1}, e_{k_2}$, and $e_{k_3}$ is always a simplex, the
real numbers $c_1, c_2$, and $c_3$ are uniquely determined. By
counting the possible indices $k_j$ and reducing the bijective
relabellings, we find that the number of inequivalent simulation
irreducible trichotomic observables equals
$\frac{1}{3}m(m-1)(m-2)$. For details of the combinatorics we refer
the interested reader to the appendix.

Combining the results for dichotomic and trichotomic observables
concludes the proof.
\end{proof}

\begin{proposition} \label{prop:poly-odd}
The minimal simulation number for the set $\obs$ of all
observables on an odd polygon state space $\mathcal{S}_{2m+1}$ equals
$\smin{\obs} = \frac{1}{6} m(m+1)(2m+1)$.
\end{proposition}
\begin{proof}
The proof follows from similar arguments as in the previous
proposition. However, for odd polygon state spaces there are no
indecomposable dichotomic observables because the extreme rays
\eqref{eq:rays} are aligned in such a way that no positive linear
combination of two effects in the extreme rays can sum up to $u$.
In other words, the complement of any indecomposable effect $c_j
g_{k_j}$ does not belong to an extreme ray. For this reason we
focus on trichotomic simulation irreducible observables $\A$ with
effects $\A_j = c_j g_{k_j}$, $j=1,2,3$. Since we are interested
in inequivalent observables $\A$, the effects $\A_1$, $\A_2$,
and $\A_3$ are linear independent, which guarantees the uniqueness
of the convex decomposition $\sum_{j=1}^3 c_j g_{k_j} = u$. The number
of such observables $\A$ is merely the number of ways to choose three
points $k_1$, $k_2$, and $k_3$ among $2m+1$ vertices of a regular
polygon with restriction that the center of the polygon belongs to the
triangle $\triangle(k_1,k_2,k_3)$. The number of different ways
equals $\frac{1}{6} m(m+1)(2m+1)$. For details of the combinatorics we
refer the interested reader to the Appendix.
\end{proof}

Propositions~\ref{prop:poly-even} and \ref{prop:poly-odd} show that
in any polygon state spaces with more than four vertices there always
exists trichotomic simulation irreducible observables. Since any
simulation irreducible observable can be simulated with its
minimally sufficient representative, which has been shown to have at
most three outcomes for polygon state spaces, we conclude that in
any polygon state space with $n \geq 5$ vertices the effective
number of outcomes for the whole space of observables $\obs$ is
exactly three.

\begin{corollary}
For any polygon state space $\mathcal{S}_n$ with $n \geq 5$ the
set of all observables is effectively trichotomic, i.e., $\obs =
\obs^{\rm eff}_3$.
\end{corollary}

Finally, the following example illustrates the effect of noise on
simulability of observables.

\begin{example}
Consider a hexagon state space $\state_6$ and a trichotomic
simulation irreducible observable $\A$ with effects $\A_1 =
\frac{2}{3} e_1$, $\A_2 = \frac{2}{3} e_3$, $\A_3 = \frac{2}{3}
e_5$, where the effects $e_k$ are given by
formula~\eqref{eq:poly-effect}.  
Obviously, $\A$ is effectively
trichotomic as it is simulation irreducible. Let us show that the
noisy observable $\A'$ with effects $\A'_{k} = (1-\lambda) \A_k +
\lambda \frac{1}{3}u$ becomes effectively dichotomic if
$\frac{1}{4} \leq \lambda < 1$. In fact, if $\lambda = \frac{1}{4}$, then $\A'_1 =
\frac{1}{3} (e_1 + \frac{1}{2}e_6 + \frac{1}{2} e_2)$, $\A'_2 =
\frac{1}{3} (e_3 + \frac{1}{2}e_2 + \frac{1}{2} e_4)$, $\A'_3 =
\frac{1}{3} (e_5 + \frac{1}{2}e_4 + \frac{1}{2} e_6)$. If this is
the case, then $\A'_k = \frac{1}{3} \sum_{i=1}^{3} \sum_{x = \pm}
\nu_{(i,x)k} \B_x^{(i)}$, where $\B^{(i)}$ is a dichotomic observable
with effects $\B^{(i)}_+ =  e_{2i-1}$ and $\B^{(i)}_- =
 e_{2i + 2}$, $i=1,2,3$ (addition in indices is modulo
6), $\nu_{(i,x)k}$ is the right stochastic matrix with elements
$\nu_{(i,+)k} = 1$ and $\nu_{(i,-)k} = 0$ if $i=k$, $\nu_{(i,+)k} = 0$
and $\nu_{(i,-)k} = \frac{1}{2}$ if $i \neq k$. Clearly,
for larger noise the observable $\A'$ remains effectively
dichotomic unless $\lambda = 1$, when the observable $\A'$
becomes trivial.
\end{example}

The above example illustrates that sufficiently noisy observables
can be simulated by dichotomic observables.

\section*{Conclusions}

Within the framework of generalized probabilistic theories, we
have considered the fundamental properties of the set of
observables $\simu{\mathcal{B}}$ that can be obtained from another
set of observables $\mathcal{B}$ via mixing and postprocessing.
Mathematically, the simulation map $\simu{\cdot}$ is an algebraic closure operator on the set of observables.
We introduced the concept of a simulation irreducible observable, which turned out to be useful in the analysis of simulability. 
In particular, we have shown that any observable can be
simulated by a finite number of simulation irreducible ones. 

The benefit of a simulation scheme is that a wide class of observables can
be realized (experimentally) via a small number of
simulators. 
We have discussed the minimal simulation number
$\smin{\mathcal{B}}$ as an indicator of the incompatibility of a subset $\mathcal{B}$ of observables, and we pointed out its connection (in the case of
quantum theory) to $k$-compatibility of observables. 
Another way to benefit from a simulation scheme is that one  can simulate observables with a larger number of outcomes as compared with the number of outcomes for simulators. 
This means that a class of observables with many outcomes can be achieved by using, e.g., dichotomic simulators, in which case we can regard those observables as effectively dichotomic.

We found that the effects of an effectively dichotomic observable have a simple geometric characterization in terms of the effects of the dichotomic simulator observables. This then serves as a useful necessary condition for dichotomic simulability when the set of available dichotomic measurement devices is fixed. We also showed that the condition becomes sufficient when we pose some additional restrictions on the simulator observables.

Finally, we have considered particular examples of nonquantum
state spaces. 
The classical state spaces are the state spaces where there exists, up to equivalence, only one simulation irreducible observable. 
In general, the number of inequivalent simulation irreducible observables is a characteristic feature of a state space.
We have considered even and odd
polygon state spaces $\mathcal{S}_n$ in detail. In contrast to
quantum theory, where there exists a continuum of inequivalent simulation
irreducible observables, in any polygon state space the minimal
simulation number for the set of all observables is finite. 
Also, we have shown
that the set of all observables is effectively dichotomic for $n=4$
and effectively trichotomic for $n \geq 5$. 
By a specific example
we have illustrated how an effectively trichotomic
observable becomes effectively dichotomic under the addition of noise.

\section*{Acknowledgements}

The authors wish to thank Martin Pl\'avala for useful discussions and Tom Bullock for useful comments on the manuscript. This work was performed as part of the Academy of Finland Centre
of Excellence program (Project No. 312058). S.N.F. acknowledges the
support of Academy of Finland for a mobility grant to conduct
research in the University of Turku. S.N.F. thanks the Russian
Foundation for Basic Research for partial support under Project
No. 16-37-60070 mol-a-dk. L.L. 
acknowledges financial support from University of Turku Graduate School.

\section{Appendix}

\subsection*{Proof of Proposition~\ref{prop:nonzero-extreme}}

\begin{proof}
Suppose the nonzero effects $\A_1,\ldots,\A_n$ of an observable
$\A$ are linearly dependent, i.e.,
\begin{equation}\label{linear-dependence}
\sum_{i=1}^{n} r_i \A_i = 0
\end{equation}
for real $r_i$ such that $\sum_{i} |r_i| > 0$. This implies that
\begin{equation}
\sum_{i: \ r_i \geq 0} r_i \A_i = \sum_{i: \ r_i < 0} |r_i| \A_i.
\end{equation}

\noindent Denote $\lambda = \frac{1}{2 \max_i |r_i|} > 0$ and
consider two observables $\B$ and $\C$ defined as follows:
\begin{eqnarray}
\B_i = \left\{ \begin{array}{ll}
  (1 - \lambda r_i) \A_i & \text{if~} r_i \geq 0, \\
  (1 + \lambda |r_i|) \A_i & \text{if~} r_i < 0, \\
\end{array}\right. \\
\C_i = \left\{ \begin{array}{ll}
  (1 + \lambda r_i) \A_i & \text{if~} r_i \geq 0, \\
  (1 - \lambda |r_i|) \A_i & \text{if~} r_i < 0. \\
\end{array}\right.
\end{eqnarray}
It is straightforward to see that $\B$ and $\C$ are indeed observables. Now it follows that
\begin{equation}
\A= \dfrac{1}{2} \B + \dfrac{1}{2}\C.
\end{equation}
Therefore, $\A$ is not extreme.
\end{proof}

\subsection*{Proof of property (sim6)}

Take $\A, \A' \in \simu{\mathcal{B}}$ so that there exists two
finite sets of observables $\{\B^{(i)}\}_{i=1}^m,
\{\B'^{(j)}\}_{j=1}^{m'} \subseteq \mathcal{B}$ with outcome sets
$X$ for $\B^{(i)}$'s and $X'$ for $\B'^{(j)}$'s, probability
distributions $\{p_i\}_{i=1}^m, \{p'_j\}_{j=1}^{m'} \subset [0,1]$
and postprocessings $\nu: \{1, \ldots,m\} \times X \to Y$ and
$\nu': \{1, \ldots,m'\} \times X' \to Y'$ for some outcome sets
$Y$ and $Y'$ such that
\begin{equation}
\A_y = \sum_{(i,x)} \nu_{(i,x)y} p_i \B^{(i)}_x, \quad \A'_{y'} = \sum_{(j,x')} \nu'_{(j,x')y'} p'_j \B'^{(j)}_{x'}
\end{equation}
for all $y \in Y$ and $y' \in Y'$.

For any $0 \leq \lambda \leq 1$ we may form a mixture of $\A$ and $\A'$ with outcome set $Y_{\rm mix} \equiv Y \cup Y'$ so that
\begin{eqnarray}
&& \lambda \A_y + (1-\lambda) \A'_y \nonumber\\
&& = \sum_{(i,x)} \nu_{(i,x)y} \lambda p_i \B^{(i)}_x +
\sum_{(j,x')} \nu'_{(j,x')y} (1-\lambda) p'_j \B'^{(j)}_{x'},
\end{eqnarray}
where we have also extended both postprocessings on $Y_{\rm mix}$ by setting $\nu_{(i,x)y}= 0$ if $y \notin Y$ and $\nu'_{(i,x)y}= 0$ if $y \notin Y'$.

We see now that we can use the observables $\{\B^{(1)}, \ldots, \B^{(m)}, \B'^{(1)}, \ldots, \B'^{(m')} \} \subseteq \mathcal{B}$ to simulate the mixture $\lambda \A + (1- \lambda) \A'$. Namely, if we denote $\B^{(m+i)} = \B'^{(i)}$ for all $i = 1, \ldots,m'$ and consider the probability distribution $\{\tilde{p}_i\}_{i=1}^{m+m'} \equiv \{\lambda p_1, \ldots, \lambda p_m,(1-\lambda) p'_1, \ldots, (1-\lambda)p'_{m'}\} \subset [0,1]$, we may define the mixture observable $\tilde{B}$ with outcome set $\{1, \ldots,m+m'\} \times X_{\rm mix}$, where $X_{\rm mix} \equiv X \cup X'$, by
\begin{equation}
\tilde{\B}_{(i,x)} = \tilde{p}_i \B^{(i)}_x
\end{equation}
for all $i=1, \ldots, m+m'$ that keeps track of the measured observable. Similarly we can define a postprocessing $\mu: \{1, \ldots,m+m'\} \times X_{\rm mix} \to Y_{\rm mix}$ by
\begin{equation}
\mu_{(i,x)y} = \chi_{\{1, \ldots,m\}}(i) \nu_{(i,x)y} + \chi_{\{m+1, \ldots,m+m'\}}(i) \nu'_{(i-m,x)y},
\end{equation}
where $\chi_S$ is the characteristic function of a set $S \subset \integer$ so that $\chi_S(x) = 1 $ if $x \in S$ and $\chi_S(x)=0$ otherwise. Now
\begin{align*}
(\mu \circ \tilde{\B})_y &= \sum_{(i,x)} \mu_{(i,x)y} \tilde{\B}_{(i,x)}  \\
&= \sum_{i=1}^m \sum_x \nu_{(i,x)y} \tilde{\B}_{(i,x)}  + \sum_{i=m+1}^{m+m'} \sum_x \nu'_{(i-m,x)y} \tilde{\B}_{(i,x)} \\
&= \sum_{i=1}^m \sum_{x} \nu_{(i,x)y} \lambda p_i \B^{(i)}_x + \sum_{j=1}^{m'} \sum_{x} \nu'_{(j,x)y} (1-\lambda) p'_j \B'^{(j)}_{x} \\
&= \lambda \A_y + (1-\lambda) \A'_y
\end{align*}
for all $y \in Y_{\rm mix}$ so that $ \lambda \A + (1- \lambda) \A' \in \simu{\mathcal{B}}$ which shows that $\simu{\mathcal{B}}$ is convex.

\subsection*{Proof of property (sim7)}

Take $\A \in \simu{\mathcal{B}}$ with an outcome set $Y$ so that
\begin{equation}
\A_y = \sum_{(i,x)} \nu_{(i,x)y} p_i \B^{(i)}_x
\end{equation}
for all $y \in Y$, some finite set of observables $\{\B^{(i)}\}_i
\subseteq \mathcal{B}$ with outcome sets $X$, some probability
distribution $\{p_i\}_i \subset [0,1]$, and some postprocessing
$\nu: \cup_k \{k \} \times X \to Y$. If now $\mu: Y \to Z$ is a
postprocessing from $Y$ to some outcome set $Z$, then
\begin{align*}
(\mu \circ \A)_z &= \sum_y \mu_{yz} \A_y \\
&= \sum_y \mu_{yz} \left( \sum_{(i,x)} \nu_{(i,x)y} p_i \B^{(i)}_x \right) \\
&= \sum_{(i,x)} \left( \sum_y \nu_{(i,x)y} \mu_{yz} \right) p_i \B^{(i)}_x \\
&= \sum_{(i,x)} \eta_{(i,x)z} p_i \B^{(i)}_x,
\end{align*}
where we have defined the postprocessing $\eta: \cup_k \{k \}
\times X \to Z$ by $\eta_{(i,x)z} = \sum_y \nu_{(i,x)y} \mu_{yz}$
for all $i$, $x \in X$, and $z \in Z$. Thus, $\mu \circ \A \in
\simu{\mathcal{B}}$.

\subsection*{Combinatorics in proof of Proposition~\ref{prop:poly-even}}

When choosing effects $\A_{1}=c_1 e_{k_1}$, $\A_{2}=c_2 e_{k_2}$, $\A_{3}=c_3 e_{k_3}$, we
cannot have $k_l = k_j +m$ for any $l\neq j$, since then from the
decomposition $u = e_{k_j} + e_{k_j +m}$ it would follow that the
remaining effect $\A_i = c_i e_{k_i}$, $i \neq j \neq l \neq i$,
is decomposable. Secondly, we cannot have $k_l = k_j \pm1$ for any
$l \neq j$ since this would force the remaining index $k_i$, $i
\neq l \neq j \neq i$, to be either $k_i = k_j+m$ or $k_i = k_j\pm
1+m$ in order for \eqref{eq:poly-conv-hull} to hold, which in turn
would lead to a violation of the previous case. Thus, by
considering possible cases for the indices $k_j$, $j=1,2,3$, such
that \eqref{eq:poly-conv-hull} holds, we see that the problem
reduces to a simple problem of combinatorics:

\begin{itemize}
\item[i)] We can choose the effect $\A_1$ to be proportional to
any nontrivial extreme effect $e_i$, where $i \in \{1,
\ldots,2m\}$ so that $\A_1$ has $2m$ possibilities.

\item[ii)] For $\A_2$ there are $2m-4$ possibilities since $\A_2$
cannot be proportional to $e_{i-1},e_i, e_{i+1}$, or $e_{i+m}$.
Thus, we have that $\A_2$ is proportional to $e_j$, where either
$j \in \{i+2, \ldots,i+m-1\}$ or $j \in \{i+m+1, \ldots, i+2m-2\}$
so that $j$ has $m-2$ possibilities in both of these cases.

\item[iii)] If $j \in \{i+2, \ldots,i+m-1\}$, the only possibility
for $\A_3$ is to be proportional to an effect $e_k$ which is
limited to be in some of the extreme rays between the complements
of $e_i$ and $e_j$ since otherwise the convex hull of
$\{e_i,e_j,e_k\}$ would not contain $u/2$. Thus, $k \in \{i+m+1,
\ldots, j+m-1\}$ and since $j=i+l$ for some $l\in
\{2,\ldots,m-1\}$ we have that $k$ has a total of $l-1$
possibilities. By the same argument, in the case when $j \in
\{i+m+1, \ldots,i+2m-2\}$, we still have $l-1$ different
possibilities, where again each $l$ represents different $j$ from
ii).
\end{itemize}

Now we can calculate the total number of different cases. As shown
above, for $\A_1$ we have $2m$ possibilities and then for $\A_2$
and $\A_3$, there are
\begin{equation}
2 \sum_{l=2}^{m-1} (l-1) = 2 \sum_{l'=1}^{m-2} l' = 2 \,
\dfrac{(m-2)(m-1)}{2} = (m-1)(m-2)
\end{equation}
different possibilities, where the multiplier $2$ came from to two
different sets of values for $j$ in ii). In order to not to
include any bijective relabellings of the effects of $\A$ we have
to take into account the different permutations of the set
$\{1,2,3\}$. Hence, the total number of inequivalent simulation irreducible
trichotomic observable equals
\begin{equation}
\dfrac{2m(m-1)(m-2)}{3!} = \dfrac{m(m-1)(m-2)}{3}.
\end{equation}

\subsection*{Combinatorics in proof of Proposition~\ref{prop:poly-odd}}

Effects $\A_{1}=c_1 g_{k_1}$, $\A_{2}=c_2 g_{k_2}$, $\A_{3}=c_3 g_{k_3}$ can be chosen as
follows:

\begin{itemize}
\item[i)] $\A_1$ is proportional to one of the nontrivial extreme
effects $f_i$, where $i \in \{1, \ldots,2m+1\}$ so that for $\A_1$
we have $2m+1$ possibilities.

\item[ii)] For $\A_2$ there are $2m$ possibilities since $\A_2$
cannot be proportional to $g_{i}$. Thus, we have that $\A_2$ is
proportional to $g_j$, where either $j \in \{i+1, \ldots,i+m\}$ or
$j \in \{i+m+1, \ldots, i+2m\}$ so that $j$ has $m$ possibilities
in both of these cases.

\item[iii)] If $j \in \{i+1, \ldots,i+m\}$, the only possibility
for $\A_3$ is to be proportional to an effect $g_k$ with $k \in
\{i+m+1, \ldots, j+m\}$ and since $j=i+l$ for some $l\in
\{1,\ldots,m\}$ we have that $k$ has a total of $l$ possibilities.
By the same argument, in the case when $j \in \{i+m+1,
\ldots,i+2m\}$, we still have $l$ different possibilities, where
again each $l$ represents different $j$ from ii).
\end{itemize}

From this we can calculate the total number of different cases. As
shown above, for $\A_1$ we have $2m+1$ possibilities and then for
$\A_2$ and $\A_3$, there are
\begin{equation}
2 \sum_{l=1}^{m} l = 2 \, \dfrac{m(m+1)}{2} = m(m+1)
\end{equation}
different possibilities, where the multiplier $2$ came from to two
different sets of values for $j$ in ii). In order to not to
include any bijective relabellings of the effects of $\A$ we have
to take into account the different permutations of the set
$\{1,2,3\}$. Hence, the total number of inequivalent simulation irreducible
trichotomic observable equals
\begin{equation}
\dfrac{(2m+1)m(m+1)}{3!}.
\end{equation}



\begin{thebibliography}{10}
\bibitem{Plavala16}
M.~Pl\'avala.
\newblock All measurements in a probabilistic theory are compatible if and only
  if the state space is a simplex.
\newblock {\em Phys. Rev. A}, 94:042108, 2016.
\end{thebibliography}
\end{document}